\documentclass [11pt,twoside,a4paper]{article}
\usepackage{amsfonts}
\usepackage{amsthm}
\usepackage{amsmath}
\usepackage{amstext}
\usepackage{amssymb}
\usepackage{mathrsfs}
\usepackage{amscd}
\usepackage{xypic}
\usepackage{graphicx}          %����ͼ�������ĺ���
\usepackage{fancybox}          %���Ǻ��������ĺ���
\usepackage{color}             %������ɫ����
\usepackage{fancyhdr}
\usepackage[hang,footnotesize]{caption2}  %�й�ͼ�α����ĺ���

\def\mathcal{\mathscr}
\newfont{\aaa}{cmb10 at 19pt}
\newfont{\bbb}{cmb10 at 11pt}

\newcommand{\beq}{\begin{equation}}
\newcommand{\eeq}{\end{equation}}
\newcommand{\bey}{\begin{eqnarray}}
\newcommand{\eey}{\end{eqnarray}}
\newcommand{\beyy}{\begin{eqnarray*}}
\newcommand{\eeyy}{\end{eqnarray*}}

\usepackage{algorithm,algorithmic}
\usepackage{verbatim}
\usepackage{subfigure}

\newtheorem{theorem}{Theorem}[section]

\newtheorem{corollary}{Corollary}[section]

\newtheorem{definition}{Definition}

\newcommand{\mcp}{\mathcal{P}}
\newcommand{\mcl}{\mathcal{L}}
\newcommand{\mbr}{\mathbb{R}}
\newcommand{\rank}{\mathrm{rank}}
\newcommand{\new}{{\mathrm{new}}}

\newcommand{\st}{\mathrm{s.t.}}

 \makeatletter
\def\@evenhead{
   \vbox{\hbox to \textwidth
{}{\hspace{0mm}{\footnotesize \thepage}}{\hspace{7cm}
   {\footnotesize {Y. Xu, W. Yin, Z. Wen, and Y. Zhang }}}
  \protect\vspace{1truemm}\relax
   \hrule depth0pt height0.15truemm width\textwidth
   }}
   \def\@evenfoot{}
\def\@oddhead{
    \vbox{\hbox to \textwidth
  {{\hspace{0cm}{\footnotesize An Alternating Direction Algorithm for Matrix Completion with Nonnegative Factors}
  \hfill{\footnotesize \thepage}}\hspace{0mm}}{}
   \protect\vspace{1truemm}\relax
   \hrule depth0pt height0.15truemm width\textwidth
  }}
  \def\@oddfoot{}
\makeatother
\begin{document}

\title{An Alternating Direction Algorithm for Matrix Completion with Nonnegative Factors
}
\date{}

\author{
Yangyang Xu$^1$,
\quad Wotao Yin$^1$,
\quad Zaiwen Wen$^2$,
\quad Yin Zhang$^1$ \\[.1in]
$^1${\small Department of Computational
and Applied Mathematics, Rice University}\\
$^2${\small Department of Mathematics and Institute of Natural Sciences,}\\
{\small Shanghai Jiaotong University}
}

\maketitle

\textbf{Abstract.} This paper introduces an algorithm for the nonnegative matrix factorization-and-completion problem, which aims to find nonnegative low-rank matrices $X$ and $Y$ so that the product $XY$ approximates a nonnegative data matrix $M$ whose elements are partially known (to a certain accuracy). This problem aggregates two existing problems: (i) nonnegative matrix factorization where all entries of $M$ are given, and (ii) low-rank matrix completion where nonnegativity is not required. By taking the advantages of both nonnegativity and low-rankness, one can generally obtain superior results than those of just using one of the two properties. We propose to solve the non-convex constrained least-squares problem using an algorithm based on the classic alternating direction augmented Lagrangian method. Preliminary convergence properties of the algorithm and numerical simulation results are presented. Compared to a recent algorithm for nonnegative matrix factorization, the proposed algorithm produces factorizations of similar quality using only about half of the matrix entries. On tasks of recovering incomplete grayscale and hyperspectral images, the proposed algorithm yields overall better qualities than those produced by two recent matrix-completion algorithms that do not exploit nonnegativity.



\section{Introduction}
This paper introduces an algorithm for the following problem:
\begin{definition}[Nonnegative matrix factorization / completion (NMFC)]\label{nmfc}
Given samples $M_{i,j}$, $(i,j)\in\Omega\subset \{1,\ldots,m\}\times \{1,\ldots,n\}$, of a \emph{nonnegative rank-$r$} matrix $M\in\mathbb{R}^{m\times n}$, find \emph{nonnegative} matrices $X\in\mathbb{R}^{m\times q}$ and $Y\in\mathbb{R}^{q\times n}$ such that $\|M-XY\|_F$ is minimized.
\end{definition}
Note that $q$ is not necessarily set to equal $r$. Firstly, not all rank-$r$ nonnegative matrices have nonnegative factors of size $r$. For some of them, the available size of  nonnegative factors is strictly greater than $r$. Secondly, when $M$ is approximately low-rank, i.e. the singular values of $M$ have a fast-decaying distribution, one often sets  $q$ to be the estimated rank or the  number of significant singular values. This resulting problem can be called \emph{approximate NMFC}. In general, depending on  data and applications, $q$ can be either equal, less than, or greater than $r$. 

NMFC is a combination of nonnegative matrix factorization (NMF) --- which finds nonnegative factors of a nonnegative matrix given all of its entries --- and low-rank matrix completion (LRMC) --- which recovers $M$ from an incomplete set of its entries without assuming nonnegativity. Mathematically, given a matrix $M\in\mathbb{R}^{m\times n}$ and  $q>0$, we present the three problems with the following models
\begin{eqnarray}\label{eqnmf}
\mbox{NMFC:} && \min_{X,Y}\left\{\|\mcp_\Omega(XY - M)\|_F^2:
\begin{array}{l} X\in\mathbb{R}^{m\times q}, Y\in\mathbb{R}^{q\times n}, \\
X_{ij}\ge 0, Y_{ij}\ge 0,\forall~i,j
\end{array}\right\},\\
\label{nmf}
\mbox{NMF:} && \min_{X,Y}\left\{\|XY - M\|_F^2 : \begin{array}{l}X\in\mathbb{R}^{m\times q}, Y\in\mathbb{R}^{q\times n},\\ X_{ij}\ge 0, Y_{ij}\ge 0,\forall~i,j\end{array}\right\},\\
\label{lrmc}
\mbox{LRMC:} && \min_Z\left\{\rank(Z): \begin{array}{l}Z\in\mathbb{R}^{m\times n}, \\ \mcp_\Omega (Z-M) = 0
\end{array}\right\},
\end{eqnarray}
where $\Omega$ indexes the known entries of $M$ and $\mcp_\Omega(A)$ returns a copy of $A$ that zeros out the entries not in $\Omega$. Note that each of the three problems has other models. Examples include weighted least-squares for NMF and NMFC and nuclear-norm minimization for LRMC. While \eqref{eqnmf} and \eqref{nmf} return $X Y$ up to a fixed rank $q$, \eqref{lrmc} seeks for a \emph{least-rank} recovery $Z$. It is well known that models \eqref{eqnmf}--\eqref{lrmc} are non-convex and generally difficult to solve. A recent advance for \eqref{lrmc} is that if $M$ is low-rank and the samples $\Omega$ satisfy the so-called incoherence property and are sufficiently large, then a convex problem based on nuclear norm minimization can exactly recover $M$ (see the pioneering work \cite{fazel-PhD-Thesis2002}, as well as recent results \cite{recht2010guaranteed, candes2009exact, wright2011robust, candes2010power}).

We are interested in NMFC since it complements NMF and LRMC. %, and it is useful when the underlying matrix has both low rank and nonnegative factors.
NMF has been widely used in data mining such as text mining, dimension reduction and clustering, as well as spectral data analysis. It started to  appear in \cite{paatero1994positive, paatero1997least, paatero1999multilinear} and has become popular since the publication of \cite{lee1999learning} in 1999. More information on NMF can be found in the survey paper \cite{berry2007algorithms}, as well as books \cite{cichocki2008advances, separationnonnegative}. Unlike NMF, NMFC assumes that the underlying matrix is incompletely sampled; hence, it leads to saving of sampling time and storage (for data such as images) and  has broader applicability.
On the other hand, LRMC has recently found a large number of applications including \emph{collaborative filtering}, which is used by Netflix to infer individual preference from an incomplete set of user preferences \cite{goldberg1992using}, \emph{global positioning}, which discovers the positions of nodes in a network from incomplete pair-wise distances \cite{biswas2006semidefinite}, \emph{system identification and order reduction}, which recovers or reduces the dimension of the state vectors of a linear time-invariant state-space model \cite{liu2009interior}, as well as the \emph{background subtraction} and \emph{structure-from-motion} problems in computer vision. A rank-$q$ matrix $M$ can be written as $M=XY$ for matrices $X$ with $q$ columns and $Y$ with $q$ rows. When $X$ and $Y$ are known to be nonnegative \emph{a priori}, empirical evidence given in Section \ref{sec:numerical} shows that imposing nonnegativity on the factors improves the recovery quality. In particular, in certain applications such as hyperspectral unmixing, the factors are nonnegative due to their physical nature, so these applications will benefit from NMFC. To summarize, NMFC combines NMF and LRMC, and NMFC is useful when the underlying matrix has both low rank and nonnegative factors.

\subsection{Related Algorithms}
There are two algorithms that have been widely used for NMF: the alternating least squares (ALS) in \cite{paatero1994positive} and multiplicative updating (Mult) in \cite{lee2001algorithms}. The former algorithm alternatively updates factor matrices $X$ and $Y$ to reduce the least-squares cost $\|XY - M\|_F^2$. The closed-form updates are given as
\begin{eqnarray*}
X_{\new} &\gets & \max\{0,MY^\top(YY^\top)^\dag\},\\
Y_\new &\gets & \max\{0,(X^\top X)^\dag X^\top M\},
\end{eqnarray*}
where $\max\{\cdot,\cdot\}$ is applied component-wise and $\dag$ denotes pseudo-inverse. The algorithm Mult has much cheaper multiplicative updates
\begin{eqnarray*}
(X_\new)_{ij} & \gets & X_{ij}(M Y^\top)_{ij}/(XYY^\top +\epsilon)_{ij}, ~\forall~i,j,\\
(Y_\new)_{ij} & \gets & Y_{ij}(X^\top M)_{ij}/(X^\top X Y + \epsilon)_{ij}, ~\forall~i,j,
\end{eqnarray*}
which do not involve matrix inversion. Starting from a nonnegative initial
matrix $Y$, $X$ and $Y$ remain nonnegative during the iterations of Mult. The
algorithm presented in this paper also applies to NMF if a complete sample set
$\Omega$ is used. The resulting algorithm, which has been studied in paper
\cite{zhang2010admnmf}, is simpler and compares favorably with ALS and Mult in
terms of both speed and solution quality. In fact, the proposed algorithm in
this paper extends the work in \cite{zhang2010admnmf}, and both algorithms are
based on the algorithm of alternating direction method of multipliers (ADM)
\cite{glowinski1975,
gabay1976dual,BertsekasTsitsiklis-book-parallel,WangYangYinZhang2008,YangZhangYin2008,AltSDP-WenGoldfarbYin-2009}. 
Likewise, we can extend the algorithms ALS and Mult to solving NMFC. Extending
ALS is as straightforward as adopting the least-square cost
$\|\mcp_\Omega(XY-M)\|_F^2$ and deriving the corresponding updates. One simple
approach to extend Mult is to replace $M$ by $\tilde M\in\mathbb{R}^{m\times
n}$, defined component-wise by $\tilde M_{ij} =
M_{ij}\mathbf{1}_{(i,j)\in\Omega}$, i.e. $\tilde M$ is a copy of $M$ with the
unsampled entries set to 0. Drawing conclusions based on the comparative results
in \cite{zhang2010admnmf}, we believe that ADM based methods deliver
higher-quality solutions in shorter times.

There are also several algorithms for LRMC. Since LRMC can complete a matrix and return factors that happen to be (approximately) nonnegative, we shall briefly review a few well-known LRMC algorithms and compare them to the proposed algorithm. Singular value thresholding (SVT) \cite{cai2008singular} and fixed-point shrinkage (FPCA) \cite{ma2009fixed} are two well-known algorithms. SVT applies the linearized Bremgan iterations \cite{Yin-Osher-Goldfarb-Darbon-07} to the unconstrained nuclear-norm model of LRMC:
\begin{equation}\label{nuc}
\min \lambda\|Z\|_* + (1/2)\|\mcp_\Omega(Z-M)\|_F^2.
\end{equation}
FPCA solves the same model using iterations based on an iterative shrinkage-thresholding algorithm \cite{Hale-Yin-Zhang-07-theory}. Furthermore, classic alternating direction augmented Lagrangian methods have been applied to solving \eqref{nuc} or its variant with constraints $\mcp_\Omega(Z-M)=0$ in \cite{Goldfarb-Ma-Wen-Allerton-09, Yang-Yuan-10}. The algorithm LMaFit \cite{Wen-Yin-Zhang-LMAFIT-10} uses a different model:
\begin{equation}\label{lmafit}
\min_{X,Y,Z}\{\|XY - Z\|_F:\mcp_\Omega(Z-M)=0\}.
\end{equation}
The model is solved by a nonlinear successive over-relaxation algorithm \cite{Grippo-Sciandrone-00}. In section 3, we compare the proposed algorithm to FPCA and LMaFit and demonstrate the benefits of taking advantages of factor nonnegativity.

\subsection{Organization}
The rest of this paper is organized as follows. Section 2 reviews the ADM algorithm and presents an ADM-based algorithm for NMFC. A preliminary convergence result of this algorithm is given in Section 2.3. Section 3 presents the results of numerical simulations, which perform tasks such as decomposing nonnegative matrices, compressing grayscale images, as well as recovering three-dimensional %synthetic and real
 hyperspectral cubes from incomplete samples. Finally, Section 4 concludes this paper.

\section{Algorithm and Convergence}
\subsection{Background: the ADM approach}
In a finite-dimensional setting, the classic alternating direction
method (ADM) solves structured convex programs in the form of
\begin{equation}\label{eqcadm}
\underset{x\in\mathcal{X}, y\in\mathcal{Y}}{\min} f(x)+g(y), \text{s.t.}~ Ax+By=c,
\end{equation}
where $f$ and $g$ are convex functions defined on closed subsets
$\mathcal{X}$ and $\mathcal{Y}$ of a finite-dimensional space,
respectively, and $A,B$ and $c$ are matrices and vector of appropriate
sizes.
The augmented Lagrangian of \eqref{eqcadm}
is
$$\mathcal{L}_A(x,y,\lambda)=f(x)+g(y)+\lambda^T(Ax+By-c)+\frac{\beta}{2}\|Ax+By-c\|_2^2,$$
where $\lambda$ is a Lagrangian multiplier vector and $\beta>0$ is a
penalty parameter.

The classic alternating direction method is an extension of the
augmented Lagrangian multiplier method \cite{hestenes1969multiplier, powell1969nonlinear,
  rockafellar1973multiplier}. It performs  minimization with
respect to $x$ and $y$ alternatively, followed by the update of $\lambda$; that is, at iteration $k$, 
\begin{subequations}
\begin{align}
x^{k+1}&\leftarrow \underset{x\in\mathcal{X}}{\arg\min}\mathcal{L}_A(x,y^k,\lambda^k)\label{eqcadm1},\\
y^{k+1}&\leftarrow\underset{y\in\mathcal{Y}}{\arg\min}\mathcal{L}_A(x^{k+1},y,\lambda^k)\label{eqcadm2},\\
\lambda^{k+1}&\leftarrow \lambda^k+\gamma\beta(Ax^{k+1}+By^{k+1}-c)\label{eqcadm3},
\end{align}
\end{subequations}
where $\gamma\in(0,1.618)$ is a step length. While
\eqref{eqcadm1} only involves $f(x)$ in the objective and
\eqref{eqcadm2} only involves $g(y)$, the classic augmented Lagrangian
method requires a  minimization of $\mathcal{L}_A(x,y,\lambda^k)$ with respect to 
$x$ and $y$ jointly, i.e., replacing \eqref{eqcadm1} and \eqref{eqcadm2} by
$$(x^{k+1},y^{k+1})\leftarrow
\underset{x\in\mathcal{X},y\in\mathcal{Y}}{\arg\min}\mathcal{L}_A(x,y,\lambda^k).$$
As the minimization couples $f(x)$ and $g(y)$, it can be much more difficult than \eqref{eqcadm1} and \eqref{eqcadm2}.
\subsection{Main Algorithm}
To facilitate an efficient use of ADM, we consider an equivalent form
of \eqref{eqnmf}:
\begin{eqnarray}\label{eqlnmf}
{\begin{array}{ll}
\min_{(U,V,X,Y,Z)}& \frac{1}{2}\|XY-Z\|_F^2\\
\st& X=U, Y=V,\\
& U\ge0, V\ge0,\\
& \mcp_\Omega(Z-M)=0,
\end{array}}
\end{eqnarray}
where $X,U\in\mbr^{m\times q}$ and $Y,V\in\mbr^{q\times n}$. 
The augmented Lagrangian of \eqref{eqlnmf} is
\begin{eqnarray*}
\mathcal{L}_A(X,Y,Z,U,V,\Lambda,\Pi)&=&\frac{1}{2}\|XY-Z\|_F^2+
\Lambda\bullet (X-U)\\
&&+\Pi\bullet(Y-V)+\frac{\alpha}{2}\|X-U\|_F^2+\frac{\beta}{2}\|Y-V\|_F^2,
\end{eqnarray*}
where $\Lambda\in\mbr^{m\times q}$, $\Pi\in\mbr^{q\times n}$ are
Lagrangian multipliers, $\alpha,\beta>0$ are penalty parameters, and $A\bullet
B:=\sum_{i,j}a_{ij}b_{ij}$ for matrices $A$ and $B$ of the same size. We deliberately leave $\mcp_\Omega(Z-M)=0$ in the constraints instead of relaxing them, so only those entries of $Z$ not in $\Omega$ are free variables.

The alternating direction method for \eqref{eqlnmf} is derived by
successively minimizing $\mathcal{L}_A$ with
respect to $X,Y,Z,U,V$, one at a time while fixing others at their
most recent values, i.e.,
\begin{eqnarray*}
X_{k+1}&=&\arg\min\mcl_A(X,Y_k,Z_k,U_k,V_k,\Lambda_k,\Pi_k),\\
Y_{k+1}&=&\arg\min\mcl_A(X_{k+1},Y,Z_k,U_k,V_k,\Lambda_k,\Pi_k),\\
Z_{k+1}&=&\underset{\mcp_\Omega(Z-M)=0}{\arg\min}\mcl_A(X_{k+1},Y_{k+1},Z,U_k,V_k,\Lambda_k,\Pi_k),\\
U_{k+1}&=&\underset{U\ge0}{\arg\min}\mcl_A(X_{k+1},Y_{k+1},Z_{k+1},U,V_k,\Lambda_k,\Pi_k),\\
V_{k+1}&=&\underset{V\ge0}{\arg\min}\mcl_A(X_{k+1},Y_{k+1},Z_{k+1},U_{k+1},V,\Lambda_k,\Pi_k),
\end{eqnarray*}
and then updating the multipliers $\Lambda$ and $\Pi$. Specifically, these steps can be
written in closed form as
\begin{subequations}\label{algadm}
\begin{align}
X_{k+1}&=(Z_kY_k^T+\alpha U_k-\Lambda_k)(Y_kY_k^T+\alpha I)^{-1},\\
Y_{k+1}&=(X_{k+1}^TX_{k+1}+\beta I)^{-1}(X_{k+1}^TZ_k+\beta V_k-\Pi_k),\\
Z_{k+1}&= X_{k+1}Y_{k+1} + \mcp_\Omega(M-X_{k+1}Y_{k+1}),\\
U_{k+1}&= \mcp_+(X_{k+1} +\Lambda_k/\alpha),\\
V_{k+1}&= \mcp_+(Y_{k+1} +\Pi_k/\beta),\\
\Lambda_{k+1}&= \Lambda_k+\gamma\alpha(X_{k+1}-U_{k+1}),\label{alglam}\\
\Pi_{k+1}&= \Pi_k+\gamma\beta(Y_{k+1}-V_{k+1}),\label{algpi}
\end{align}
\end{subequations}
where $\gamma\in(0,1.618)$ and
$(\mcp_+(A))_{ij}=\max\{a_{ij},0\}$. Since matrix inversions are applied to $q\times q$ matrices, they are relatively
inexpensive for $q< \min\{m,n\}$.
\subsection{Convergence}
Global convergence can be obtained when the classic ADM is applied to two-block convex programs in the form of \eqref{eqcadm}. However, to the best of our knowledge, there is no global convergence result in general for non-convex programs or convex programs with three or more blocks. Note that problem \eqref{eqlnmf} is non-convex and there are three blocks in updates \eqref{algadm}. Due to these difficulties, we provide a convergence property of the proposed ADM algorithm that holds only under some assumptions. 

A point $(X,Y,Z,U,V)$ satisfies the KKT condition for problem \eqref{eqlnmf} if there exist $\Lambda$ and $\Pi$ such that
\begin{subequations}\label{eqkkt}
\begin{align}
(XY-Z)Y^\top+\Lambda=0,\label{eqkkt1}\\
X^\top(XY-Z)+\Pi=0,\label{eqkkt2}\\
\mcp_{\Omega^c}(XY-Z)=0,\label{eqkkt3}\\
\mcp_\Omega(Z-M)=0,\label{eqkkt4}\\
X-U=0,\label{eqkkt5}\\
 Y-V=0,\label{eqkkt6}\\
\Lambda\le0\le U, \Lambda\odot U=0,\label{eqkkt7}\\
\Pi\le0\le V, \Pi\odot V=0,\label{eqkkt8}
\end{align}
\end{subequations}
where $\Omega^c$ indexes the unobserved entries of $M$,
and $\odot$ denotes component-wise multiplication. To simplify notation, we consolidate all the variables in problem \eqref{eqlnmf} as $W:=(X,Y,Z,U,V),$
and write $\mcl_A(X)$ to represent lagrangian function with respect to $X$ by fixing others at their most recent values.
\begin{theorem}\label{thmkkt}
Let $\{(W_k,\Lambda_k,\Pi_k)\}$ be a sequence generated by the ADM algorithm
\eqref{algadm}. If the multiplier sequence $\{(\Lambda_k,\Pi_k)\}$ is bounded and satisfies
\begin{equation}\label{eqconvg}
\sum_{k=0}^\infty\left(\|\Lambda_{k+1}-\Lambda_k\|_F^2+\|\Pi_{k+1}-\Pi_k\|_F^2\right)<\infty.
\end{equation}
Then any accumulation point of $\{W_k\}$ satisfies the KKT condition for problem
\eqref{eqlnmf}. Consequently, any accumulation point of
$\{(X_k,Y_k)\}$ satisfies the KKT condition for problem \eqref{eqnmf}.
\end{theorem}

\begin{proof}
First, we claim $W_{k+1}-W_k\to 0$, and $(\Lambda_{k+1},\Pi_{k+1})-(\Lambda_{k},\Pi_{k})\to 0$. We begin the proof of this claim by observing that $\mcl_A(W,\Lambda,\Pi)$ is bounded below. This follows from
\begin{eqnarray*}
\mathcal{L}_A(W,\Lambda,\Pi)&=&\frac{1}{2}\|XY-Z\|_F^2
+\frac{\alpha}{2}\|X-U+\Lambda/\alpha\|_F^2-\frac{1}{2\alpha}\|\Lambda\|_F^2\\
&&+\frac{\beta}{2}\|Y-V+\Pi/\beta\|_F^2-\frac{1}{2\beta}\|\Pi\|_F^2,
\end{eqnarray*}
and the boundedness of $\{(\Lambda,\Pi)\}$.
Furthermore, the lagrangian function $\mcl_A$ is strongly convex with respect to each variable of $X,Y,Z,U$ and $V$. For $X$-variable, it holds for any $X$ and $\Delta X$ that
\begin{equation}\label{eq:conx}\mcl_A(X+\Delta X)-\mcl_A(X)\ge\partial_X\mcl_A(X)^\top\Delta X+\alpha\|\Delta X\|_F^2.
\end{equation}
In addition, $X^*$ is a minimizer of $\mcl_A(X)$ implies the inequality 
\begin{equation}\label{eq:vi}\partial_X\mcl_A(X^*)^\top\Delta X\ge0.
\end{equation} 
Combining \eqref{eq:conx} and \eqref{eq:vi} and observing $X_{k+1}$ is a minimizer of $\mcl_A(X)$ at the $k$-th iteration, we have
\begin{equation}\label{eq:bx}\mcl_A(X_k)-\mcl_A(X_{k+1})\ge\alpha\|X_k-X_{k+1}\|_F^2,
\end{equation}
and in the same way,
\begin{subequations}\label{eq:bw}
\begin{align}
\mcl_A(Y_k)-\mcl_A(Y_{k+1})&\ge\beta\|Y_k-Y_{k+1}\|_F^2,\\
\mcl_A(Z_k)-\mcl_A(Z_{k+1})&\ge\|Z_k-Z_{k+1}\|_F^2,\\
\mcl_A(U_k)-\mcl_A(U_{k+1})&\ge\alpha\|U_k-U_{k+1}\|_F^2,\\
\mcl_A(V_k)-\mcl_A(V_{k+1})&\ge\beta\|V_k-V_{k+1}\|_F^2.
\end{align}
\end{subequations}
Let $c:=\min\{\alpha,\beta,1\}$. Then by \eqref{eq:bx} and \eqref{eq:bw}, we have
\begin{align*}
&\mcl_A(W_k,\Lambda_k,\Pi_k)-\mcl_A(W_{k+1},\Lambda_{k+1},\Pi_{k+1})\\
=&\mcl_A(W_k,\Lambda_k,\Pi_k)-\mcl_A(W_{k+1},\Lambda_{k},\Pi_{k})\\&+\mcl_A(W_{k+1},\Lambda_k,\Pi_k)-\mcl_A(W_{k+1},\Lambda_{k+1},\Pi_{k+1})\\
\ge&c\|W_k-W_{k+1}\|_F^2-\frac{1}{\gamma\alpha}\|\Lambda_k-\Lambda_{k+1}\|_F^2-\frac{1}{\gamma\beta}\|\Pi_k-\Pi_{k+1}\|_F^2\\
\ge&c\|W_k-W_{k+1}\|_F^2-\frac{1}{c\gamma}\left(\|\Lambda_k-\Lambda_{k+1}\|_F^2+\|\Pi_k-\Pi_{k+1}\|_F^2\right).
\end{align*}
Taking summation of the above inequality and recalling $\mcl_A(W,\Lambda,\Pi)$ is bounded below, we get
$$\sum_{k=0}^\infty c\|W_k-W_{k+1}\|_F^2-\sum_{k=0}^\infty\frac{1}{c\gamma}\left(\|\Lambda_k-\Lambda_{k+1}\|_F^2+\|\Pi_k-\Pi_{k+1}\|_F^2\right)<\infty.$$
Since the second term on the left of the above inequality is bounded, it follows that
$$\sum_{k=0}^\infty c\|W_k-W_{k+1}\|_F^2<\infty,$$
from which we immediately have $W_{k+1}-W_k\to 0$. For $(\Lambda_{k+1},\Pi_{k+1})-(\Lambda_{k},\Pi_{k})\to 0$, it directly follows from \eqref{eqconvg}.

Now, we are ready to prove the result of this theorem. First, rearrange the ADM formulas in \eqref{algadm} into
\begin{subequations}\label{algradm}
\begin{eqnarray}
(X_{k+1}-X_k)(Y_kY_k^T+\alpha I)&=&-((X_kY_k-Z_k)Y_k^T\label{radm1}\\&&\quad+\alpha(X_k-U_k)+\Lambda_k),\nonumber\\
(X_{k+1}^TX_{k+1}+\beta I)(Y_{k+1}-Y_k)&=&-(X_{k+1}^T(X_{k+1}Y_k-Z_k)\label{radm2}\\&&\quad+\beta(Y_k-V_k)+\Pi_k),\nonumber\\
U_{k+1}-U_k&= &\mcp_+(X_{k+1} +\Lambda_k/\alpha)-U_k,\label{radm3}\\
V_{k+1}-V_k&= &\mcp_+(Y_{k+1} +\Pi_k/\beta)-V_k,\label{radm4}\\
\Lambda_{k+1}-\Lambda_k&= &\gamma\alpha(X_{k+1}-U_{k+1}),\label{radm5}\\
\Pi_{k+1}-\Pi_k&= &\gamma\beta(Y_{k+1}-V_{k+1}),\label{radm6}
\end{eqnarray}
\end{subequations}
and
\begin{equation}\label{admz}
Z_{k+1}= X_{k+1}Y_{k+1} + \mcp_\Omega(M-X_{k+1}Y_{k+1}).
\end{equation}
Note $W_{k+1}-W_k\to 0$, $\Lambda_{k+1}-\Lambda_k\to 0$ and $\Pi_{k+1}-\Pi_k\to 0$ imply that the left- and right-hand sides
in \eqref{algradm} all go to zero, i.e.,
\begin{subequations}\label{limadm}
\begin{eqnarray}
(X_kY_k-Z_k)Y_k^T+\Lambda_k&\to &0,\label{limadm1}\\
X_{k}^T(X_{k}Y_k-Z_k)+\Pi_k)&\to& 0,\label{limadm2}\\
\mcp_+(X_{k} +\Lambda_k/\alpha)-U_k&\to& 0,\label{limadm3}\\
\mcp_+(Y_{k} +\Pi_k/\beta)-V_k&\to &0,\label{limadm4}\\
X_{k}-U_{k}&\to& 0,\label{limadm5}\\
Y_{k}-V_{k}&\to& 0,\label{limadm6}
\end{eqnarray}
\end{subequations}
where the terms $\alpha(X_k-U_k)$ and $\beta(Y_k-V_k)$ have been
eliminated in \eqref{limadm1} and \eqref{limadm2}, respectively, by
invoking \eqref{limadm5} and \eqref{limadm6}. For any limit point $\hat{W}=(\hat{X},\hat{Y},\hat{Z},\hat{U},\hat{V})$ of sequence $\{W_k\}$, there exists subsequence $\{W_{n_k}\}$ converging to $\hat{W}$. The boundedness of $\{(\Lambda_k,\Pi_k)\}$ implies the existence of a sub-subsequence $\{(\Lambda_{n_{k_j}},\Pi_{n_{k_j}})\}$ of $\{(\Lambda_{n_{k}},\Pi_{n_{k}})\}$ converging to some point $(\hat{\Lambda},\hat{\Pi})$. Hence, $(\hat{W},\hat{\Lambda},\hat{\Pi})$ is a limit point of sequence $\{(W_k,\Lambda_k,\Pi_k)\}$. Since \eqref{admz}
exactly means
$$\mcp_\Omega(Z_k-M)=0,\text{ and }\mcp_\Omega(X_kY_k-Z_k)=0,$$ then
clearly, the first six equations in the KKT conditions \eqref{eqkkt}
are satisfied at the limit point $(\hat{W},\hat{\Lambda},\hat{\Pi})$.
The nonnegativity of $\hat{U}$ and $\hat{V}$ are guaranteed by the
algorithm construction. Therefore, we only need to verify the
non-positivity of $\hat{\Lambda}$ and $\hat{\Pi}$, and the
complementarity between $\hat{U}$ and $\hat{\Lambda}$, and between
$\hat{V}$ and $\hat{\Pi}$. Now we examine the following two
equations derived from \eqref{limadm3} and \eqref{limadm4},
respectively,
\begin{subequations}\label{limhadm}
\begin{align}
\mcp_+(\hat{X} +\hat{\Lambda}/\alpha)=\hat{U},\label{limhadm1}\\
\mcp_+(\hat{Y} +\hat{\Pi}/\beta)=\hat{V}.\label{limhadm2}
\end{align}
\end{subequations}
Note we have $\hat{X}=\hat{U}\ge0$. If
$\hat{U}_{ij}=\hat{X}_{ij}=0$, then \eqref{limhadm1} reduces
$\mcp_+(\hat{\Lambda}/\alpha)_{ij}=0$, which implies
$\hat{\Lambda}_{ij}\le 0$. On the other hand, if
$\hat{U}_{ij}=\hat{X}_{ij}>0$, then \eqref{limhadm1} implies
$\hat{\Lambda}_{ij}= 0$. This proves the non-positivity of
$\hat{\Lambda}$ and the complementarity between $\hat{U}$ and
$\hat{\Lambda}$. The same argument can be applied to
\eqref{limhadm2}, due to the identical structure, to prove the
non-positivity of $\hat{\Pi}$ and the complementarity between
$\hat{V}$ and $\hat{\Pi}$.

We have verified the statement concerning the sequence $\{W_k\}$ and
problem \eqref{eqlnmf}. The statement concerning the sequence
$\{(X_k,Y_k)\}$ and problem \eqref{eqnmf} follows directly from the
equivalence between the two problems. This completes the proof.
\end{proof}

From the proof of Theorem \ref{thmkkt}, we can immediately get the following corollary.
\begin{corollary}\label{corkkt}
Let $\{(W_k,\Lambda_k,\Pi_k)\}$ be a sequence generated by the ADM algorithm
\eqref{algadm}. Whenever the sequence converges, the limit satisfies the KKT conditions.
\end{corollary}

\section{Numerical Results}\label{sec:numerical}

\subsection{Implementation and Parameters}
A pseudo code for the proposed algorithm is given in Algorithm \ref{alg:ADM} below.
\begin{algorithm}
\caption{ADM-based algorithm for NMFC } \label{alg:ADM}
\begin{algorithmic}
\STATE Input $A=\mcp_\Omega(M)\in\mbr^{m\times
n}$, integer $q>0$, \emph{maxiter} $>0$, and $tol>0$.
\STATE Set $\alpha, \beta,
\gamma>0$. Set $Y$ as a nonnegative random matrix, $Z=A$, and
$U,V,\Lambda, \Pi$ as zero matrices of appropriate sizes.
\FOR{$k=1,\ldots, maxiter$}
\STATE Update $(X_k,Y_k,Z_k,U_k,V_k,\Lambda_k, \Pi_k)$ by the
formulas \eqref{algadm};
\IF{a stopping criterion is met}
\STATE exit and output $(X_k,Y_k)$
\ENDIF
\ENDFOR
\end{algorithmic}
\end{algorithm}

The most important  parameters are $\alpha,\beta$ and
$\gamma$. In our implementation, we set $\gamma=1.618$, and
$\beta=n\alpha/m$. The setting $\beta=n\alpha/m$ considers the different sizes of $X$ and $Y$ and balances the penalties for constraints $X=U$ and $Y=V$. The naive setting $\alpha=\beta$ also works for our tests but reduces the speed of convergence. By running a range of numerical experiments, we
heuristically scale $A$ so that $\|A\|_F=2.5\times 10^5$ and select $\alpha=2.0\times10^{-4}\|A\|_F\max(m,n)/q$. They have worked
well for our tested matrices,  and it is worth mentioning that algorithm \ref{alg:ADM} can work well for different $\alpha,\beta$ in a fairly large interval. The iteration stops once either one of the following conditions is met:
\begin{subequations}\label{eq:stopcriterion}
\begin{eqnarray}
\frac{|f_{k+1}-f_k|}{\max(1,|f_k|)}\le tol,\label{crit1}\\
f_k\le tol\label{crit2},
\end{eqnarray}
\end{subequations}
where $f_k=\|\mcp_\Omega(X_kY_k-A)\|_F/\|A\|_F$. %Moreover, we require
%that the first condition above must be satisfied at three consecutive
%iterations.
All tests were performed on a Lenovo T410 laptop with an i7-620m CPU and 3 gigabytes of memory and running 32-bit Windows 7 and MATLAB 2010b.

\subsection{Random Nonnegative Matrices Factorization}
We compared  the algorithm proposed in \cite{zhang2010admnmf} with the proposed algorithm \ref{alg:ADM}, where the former algorithm takes complete samples of a random matrix $M$ while the latter algorithm takes  75\%, 50\%, and 25\% samples of the same matrix $M$. %The other algorithm only takes full sampled matrices, so its input has sample rate 100\%.
While other reported tests in this paper used parameters and stopping rules given above, this test set used different but consistent parameters (which are not optimal for algorithm 1) for both algorithms in order to accurately reveal their performance difference and the difference between NMF and NMFC: $\alpha=\beta=10^4$ and $tol = 10^{-6}$.  We generated each rank-$r$ nonnegative matrix $M\in\mathbb{R}^{m\times n}$ in the form of $M=LDR$, where $L\in\mathbb{R}^{m\times r}$ and $R\in\mathbb{R}^{r\times n}$ were generated by calling MATLAB's command \verb!rand! and $D$ is an $r\times r$ diagonal matrix with diagonal elements $1,2,\ldots, r$. Such scaling makes $M$ slightly ill-conditioned. We tested different combinations of $n$ and $m$ and obtained roughly consistent results. Figure \ref{fg:rand_mtx} depicts the recovery qualities and speeds corresponding to $m=n=500$ and varying $q=r = 20$ through $50$. The results are the averages of 50 independent trials.

\begin{figure}
\begin{center}
\includegraphics[width=0.45\textwidth]{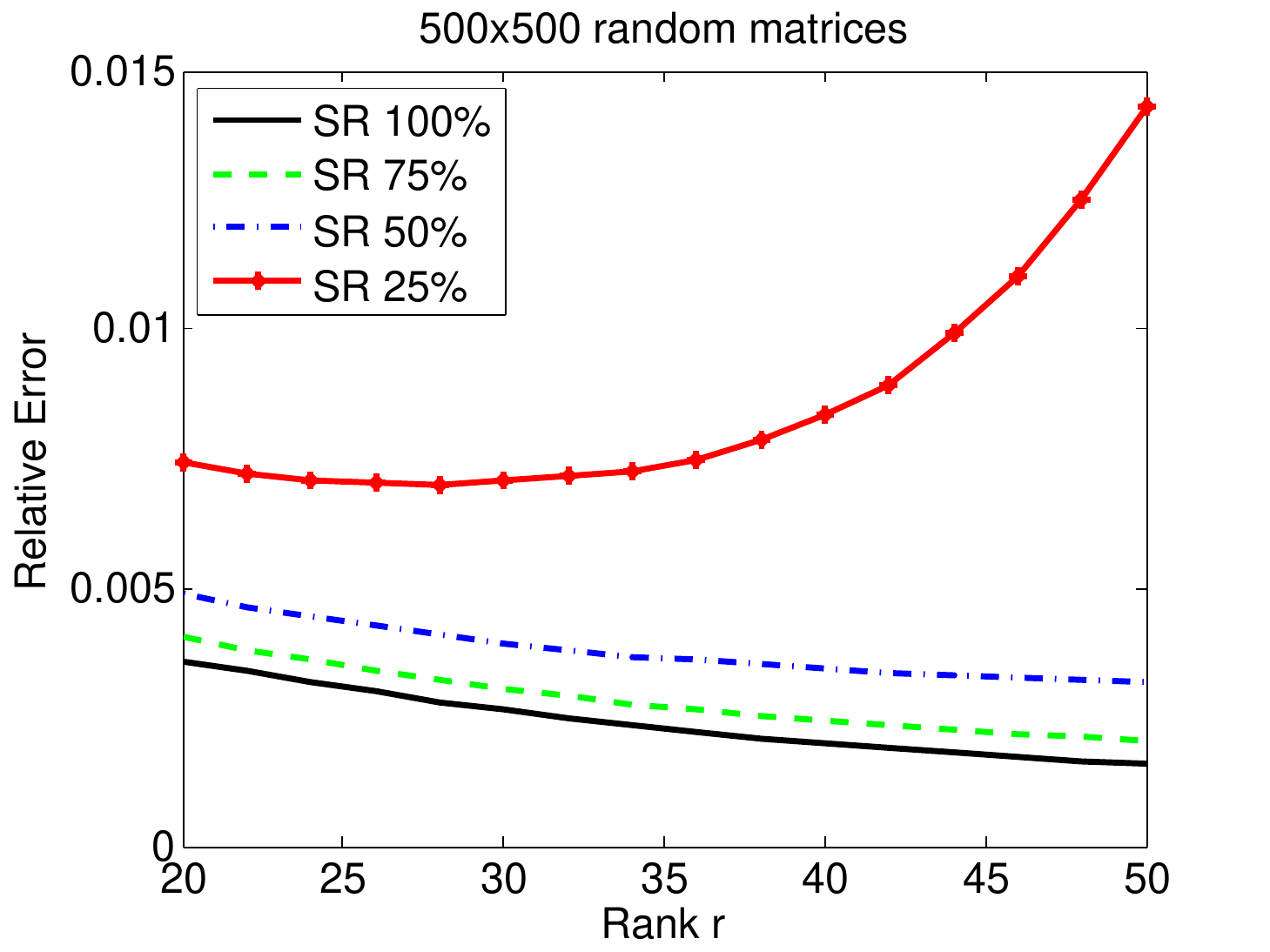}
\includegraphics[width=0.45\textwidth]{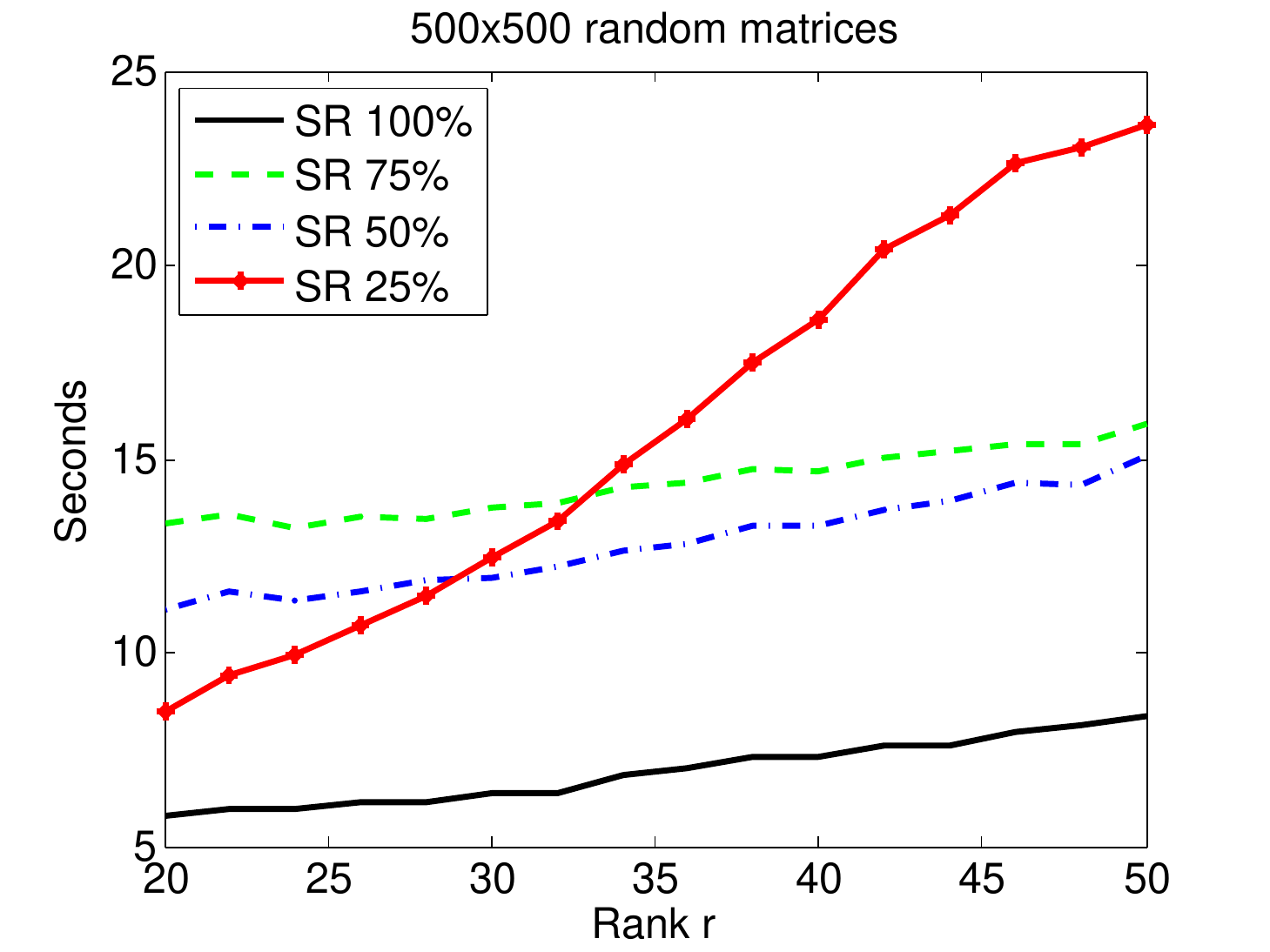}
\end{center}
\caption{Matrix completion with different sample rates (SRs). Left: relative error in Frobenious norm; Right: cpu time in seconds. The algorithm in \cite{zhang2010admnmf} was used for SR=100\%. Algorithm \ref{alg:ADM} was used for SR=70\%, 50\%, 25\%. All tests used the same parameters and stopping tolerances, and results are the averages over 50 independent trials}
\label{fg:rand_mtx}
\end{figure}

The quality of recovery is similar for SR = 100\%, 75\%, and 50\% for the set of tested matrices. They are all faithful recoveries with relative errors around 0.4\%. The relative errors for SR = 75\%, and 50\% are just slightly worse. The low SR = 25\% makes the recovery more difficult. When the ranks $r$ are between 20 and 30, the four error curves are roughly parallel though the red curve (SR = 25\%) is worse at relative errors around 0.6\%. When $r>30$, 25\% of entries seem no longer enough for faithful recovery and consequently, the red curve (SR = 25\%) begins to deviate from the others as $r$ increases, and it exhibits a steep upward trend. The difficulty with SR = 25\% samples for large $r$ is also shown in terms of cpu seconds. The times for SR = 75\% and 50\% are about three times as long as those for SR = 100\%. Since the times are the averages of merely 50 trials, the curves are not as smooth as they would be if the trials were much more.

The large gap between the red curve in Figure \ref{fg:rand_mtx}(left) and the other curves is largely due to the use of the same stopping tolerance $10^{-6}$. However, SR=25\% can reach the similar accuracy of higher SRs if it has a tighter tolerance (e.g., $10^{-7}$) and runs more iterations, at least for $r \le 30$. In this sense, lower SRs do not necessarily mean much larger errors.

\subsection{Overview of Algorithm LMaFit and FPCA}
Before more simulation results are presented, let us overview LMaFit and FPCA, which were compared to Algorithm \ref{alg:ADM} in the next two simulations.
 LMaFit solves \eqref{lmafit} based on a nonlinear successive over-relaxation (SOR) method. From its first-order optimality conditions
\begin{eqnarray*}
{\begin{array}{r}
(XY-Z)Y^\top=0,\\
X^\top(XY-Z)=0,\\
\mcp_{\Omega^c}(Z-XY)=0,\\
\mcp_\Omega(Z-M)=0.
\end{array}}
\end{eqnarray*}
 the nonlinear SOR scheme is derived as
\begin{eqnarray*}
{\begin{array}{rcl}
X_{k+1}&=&Z_kY_k^\top(Y_kY_k^T)^\dagger,\\
X_{k+1}(\omega)&=&\omega X_{k+1}+(1-\omega)X_k,\\
Y_{k+1}&=&(X_{k+1}(\omega)^\top X_{k+1}(\omega))^\dagger(X_{k+1}(\omega)^\top Z_k),\\
Y_{k+1}(\omega)&=&\omega Y_{k+1}+(1-\omega)Y_k,\\
Z_{k+1}(\omega)&=&X_{k+1}(\omega)Y_{k+1}(\omega)+\mcp_\Omega(M-X_{k+1}(\omega)Y_{k+1}(\omega)),
\end{array}}
\end{eqnarray*}
where the weight $\omega\ge 1$. One of its stopping criterions is the same as \eqref{crit1}. In our tests described below, we set tol = $10^{-5}$ for Alg \ref{alg:ADM} and LMaFit and chose different maximum numbers of iterations based on the size of recovered matrix, which will be specified below. We  applied the rank-estimation technique coming with LMaFit (hence, we did not fix $q$ for LMaFit).%, and the values of parameters related to rank-modification will be specified in each test.

FPCA  solves convex problems in the form of
\begin{equation*}
\min \mu\|X\|_*+\frac{1}{2}\|\mathcal{A}(X)-b\|_2^2,
\end{equation*}
which includes \eqref{nuc} as a special case by setting the linear operator $\mathcal{A}$ to  $\mcp_\Omega$. Introducing
$h(X)=\mathcal{A}^*(\mathcal{A}(X)-b)$, where $\mathcal{A}^*$ is
the adjoint of $\mathcal{A}$, we can write the iteration of FPCA as
\begin{eqnarray*}\
\left\{\begin{array}{rcl} Y_k&\gets &X_k-\tau h(X_k),\\
X_{k+1}&\gets &S_{\tau\mu}(Y_k), \end{array}\right.
\end{eqnarray*}
where $S_\nu(\cdot)$ is a matrix singular-value shrinkage operator.
In our tests described below, the parameters for FPCA were set to their default values: specifically, tol = $10^{-6}$ and maxiter = $10^5$. %since every iteration of FPCA is much cheaper than one iteration of algorithm \ref{alg:ADM} or LMaFit.
For the default values of other parameters such as $\tau$ and $\mu$, we refer the reader to \cite{ma2009fixed}.

\subsection{Hyperspectral Data Recovery}
In this subsection, we compare Algorithm \ref{alg:ADM} with LMaFit \cite{Wen-Yin-Zhang-LMAFIT-10} and FPCA \cite{ma2009fixed} on
recovering %synthetic and
  three-dimensional hyperspectral images from their incomplete observations. Hyperspectral (or multispectral) imaging is widely used in applications from environmental studies and biomedical imaging to military surveillance. A hyperspectral image is a three-dimensional datacube that records the electromagnetic reflectance of a scene at varying wavelengths, from which different materials in the scene can be identified by exploiting their electromagnetic scattering patterns. We let each hyperspectral datacube be represented by a three-dimensional array whose first two dimensions are spatial and third dimension is wavelength. A hyperspectral datacube can have several hundreds of wavelengths (along the third dimension) but no more than a dozen dominant materials. As a consequence, the spectral vector at every spatial location can be (approximately) linearly expressed by a small set of common vectors, called endmembers or spectral signatures of materials. The number of these basic vectors is much smaller than the number of wavelengths. Since endmembers are naturally nonnegative, a hyperspectral datacube is a set of  nonnegative mixtures of a few endmembers, which are also nonnegative. This property makes it possible to recover the endmembers and mixture coefficients from a hyperspectral datacube, and it is called \emph{unmixing}. Although unmixing is not as simple as NMF, the results of NMF can be used as an initial guess. Compared to NMF, NMFC not only performs initial unmixing but also recovers the datacube from an incomplete set of observed voxels. This advantage will translate to shorter sampling times and perhaps simpler designs of hyperspectral imaging devices.

In our simulation, the hyperspectral datacube has 163 wavelengths or slices, and the size of each slice is $80\times 80$.  Three selected slices are shown in figure \ref{fig:real-orig}. They depict an urban area at three different wavelengths. Roads, roofs,  plants, as well as other objects exhibit different intensities.  Our simulation begin with reshaping the $80\times 80\times 163$ hyperspectral datacube  to a $6400\times 163$  matrix $M$, each slice becoming one column of $M$. While $M$ is full rank, its singular values are fast decaying. We chose the estimate rank $q=30$, and set tol=$10^{-5}$ and maxiter = 2000 for Algorithm \ref{alg:ADM}, and
tol=$10^{-5}$ and maxiter = 2000, est\_rank=2, rk\_inc =3 for  LMaFit. The  parameters for FPCA were
set to their default values.

\begin{figure}[htbp]
\begin{center}
\includegraphics[width=0.7\textwidth, height=0.25\textwidth]{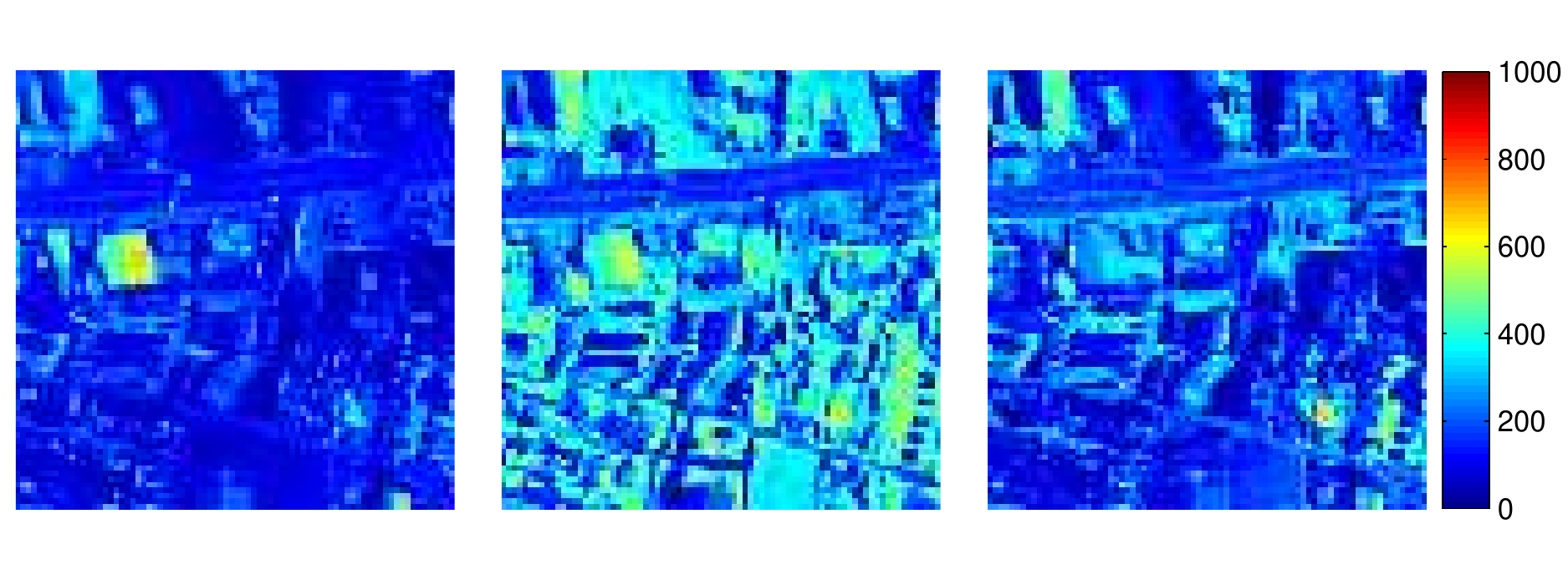}
\caption{Original slices of hyperspectral cube} \label{fig:real-orig}
\end{center}
\end{figure}

The three algorithms were compared on recovering $M$ from incomplete observations of SR = 30\%, 40\%, 50\%, and their results were compared in terms of  peak signal-to-noise ratio (PSNR), as well as mean squared error (MSE). Specifically, given a recovered matrix $\hat{M}$ from incomplete samples of $M\in\mathbb{R}^{m\times n}$, we let
\begin{align*}
\text{MSE}&:=\frac{1}{mn}\|\hat{M}-M\|_F^2,\\
\text{PSNR}&:=20\log_{10}\left(\frac{\text{MAX}_\text{I}}{\sqrt{\text{MSE}}}\right),
\end{align*}
where $\text{MAX}_\text{I}$ is the maximum  pixel intensity, which is $1023$ in this subsection for the tested hyperspectral data and $1$ in subsection \ref{subsec:imagetest} for two grayscale images. The results are listed in table \ref{table:real}, and the  three  slices of the recovered datacube that correspond to those in figure \ref{fig:real-orig} are depicted in figure \ref{fig:real-recslic}. The results show that Algorithm \ref{alg:ADM} performs
better than FPCA in both CPU time and recovery quality.  LMaFit is comparable with algorithm \ref{alg:ADM} in terms of speed but less accurate. We believe that the use of nonnegativity is a major factor for the superiority of the results of algorithm \ref{alg:ADM}.

\begin{table} \caption{real data: recovered slices by Algorithm \ref{alg:ADM},
  LMaFit, and FPCA. The rank estimate for Algorithm \ref{alg:ADM} and LMaFit is 30. }\label{table:real}
    {\scriptsize
 \setlength{\tabcolsep}{2pt}  %\centering
 \begin{center}
\begin{tabular}{|c|ccc|ccc|ccc|} \hline
\multicolumn{1}{|c|}{problem} & \multicolumn{3}{|c|}{Alg \ref{alg:ADM}}  &  \multicolumn{3}{|c|}{LMaFit} &\multicolumn{3}{|c|}{FPCA}   \\ \hline
seed & CPU & PSNR & MSE & CPU & PSNR & MSE & CPU & PSNR & MSE \\ \hline
\multicolumn{10}{|c|}{SR: 30\%}\\ \hline
  3445 &  27.15 & 47.71 & 1.77e+001  &  14.80 & 45.05 & 3.27e+001 &  39.39 & 43.31 & 4.89e+001 \\ \hline 
 31710 &  26.38 & 47.52 & 1.85e+001 &  31.35 & 43.08 & 5.15e+001 &  39.19 & 43.56 & 4.61e+001 \\ \hline 
 43875 &  27.45 & 47.50 & 1.86e+001 &  40.79 & 42.26 & 6.23e+001 &  38.49 & 44.35 & 3.84e+001 \\ \hline 
 69483 &  25.66 & 47.71 & 1.77e+001 &  42.07 & 42.20 & 6.31e+001 &  39.04 & 44.20 & 3.98e+001 \\ \hline 
 95023 &  25.67 & 47.48 & 1.87e+001 &  32.21 & 43.02 & 5.22e+001 &  39.14 & 43.13 & 5.10e+001 \\ \hline 
\hline
\multicolumn{10}{|c|}{SR: 40\%}\\ \hline
  3445 &  28.46 & 48.89 & 1.35e+001 &  27.51 & 44.96 & 3.34e+001 &  42.85 & 44.92 & 3.37e+001 \\ \hline 
 31710 &  28.66 & 49.00 & 1.32e+001 &  25.99 & 45.70 & 2.82e+001&  42.96 & 44.92 & 3.37e+001 \\ \hline 
 43875 &  29.81 & 48.88 & 1.36e+001 &  21.89 & 45.09 & 3.24e+001 &  43.32 & 44.48 & 3.73e+001 \\ \hline 
 69483 &  28.24 & 48.86 & 1.36e+001 &  21.38 & 45.26 & 3.12e+001 &  43.99 & 44.26 & 3.92e+001 \\ \hline 
 95023 &  29.10 & 48.73 & 1.40e+001 &  19.31 & 45.95 & 2.66e+001 &  43.87 & 44.79 & 3.47e+001 \\ \hline 
\hline
\multicolumn{10}{|c|}{SR: 50\%}\\ \hline
  3445 &  30.71 & 49.73 & 1.11e+001 &  34.78 & 44.69 & 3.55e+001 &  47.62 & 44.43 & 3.77e+001 \\ \hline 
 31710 &  30.39 & 49.92 & 1.07e+001 &  22.75 & 46.21 & 2.50e+001&  46.64 & 43.69 & 4.47e+001 \\ \hline 
 43875 &  31.34 & 49.74 & 1.11e+001 &  22.69 & 46.47 & 2.36e+001 &  46.77 & 44.20 & 3.98e+001 \\ \hline 
 69483 &  30.36 & 49.98 & 1.05e+001 &  29.84 & 45.00 & 3.31e+001 &  47.15 & 44.15 & 4.02e+001 \\ \hline 
 95023 &  30.27 & 49.83 & 1.09e+001 &  47.71 & 44.09 & 4.08e+001 &  47.63 & 44.36 & 3.84e+001  \\ \hline 
\end{tabular}
  \end{center}
}
 \end{table}

\begin{figure}\caption{Recovered slices by Algorithm \ref{alg:ADM} (first rows), LMaFit (second rows), and FPCA (last rows), respectively; rank estimate for Algorithm \ref{alg:ADM} and LMaFit is 30}
\label{fig:real-recslic}
\begin{minipage}[t]{1\textwidth}
\centering
\subfigure[SR = 30\%]{
\begin{minipage}[t]{0.31\textwidth}
\centering
 \includegraphics[width=0.99\textwidth, height =
 0.9\textwidth]{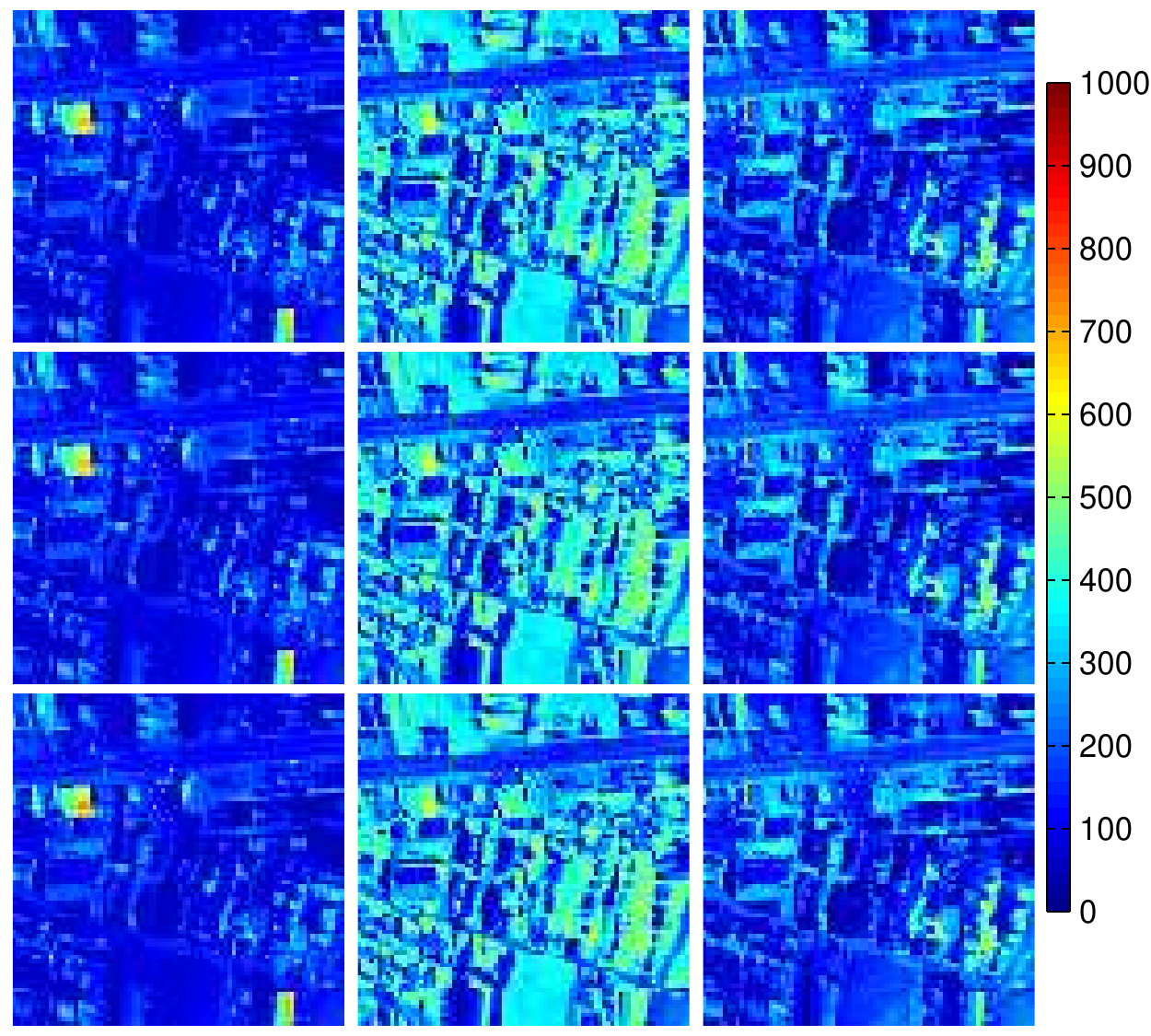}
\end{minipage}}
\subfigure[SR = 40\%]{
\begin{minipage}[t]{0.31\textwidth}
\centering
 \includegraphics[width=0.99\textwidth, height =
 0.9\textwidth]{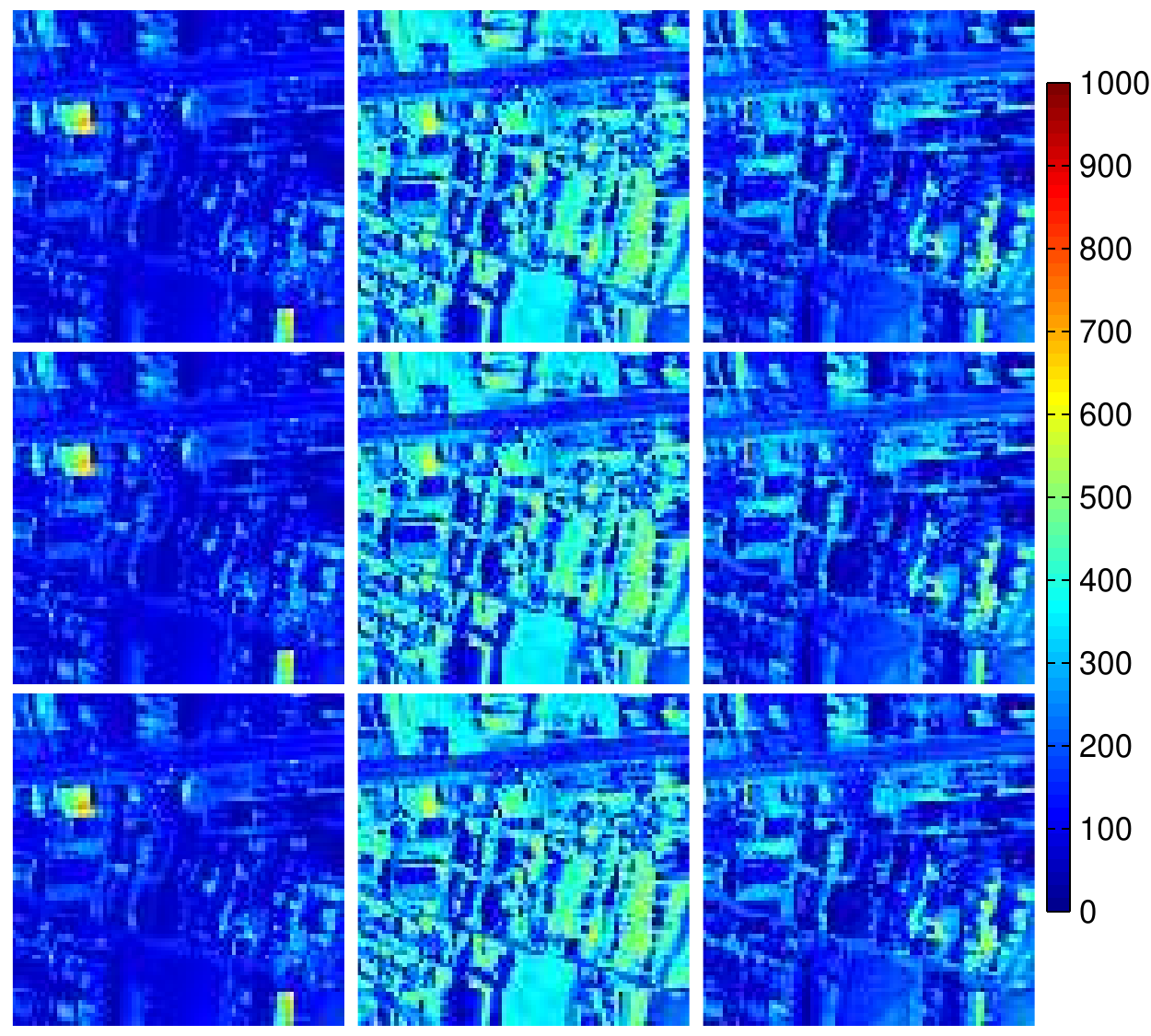}
\end{minipage}}
\subfigure[SR = 50\%]{
\begin{minipage}[t]{0.31\textwidth}
\centering
 \includegraphics[width=0.99\textwidth, height =
 0.9\textwidth]{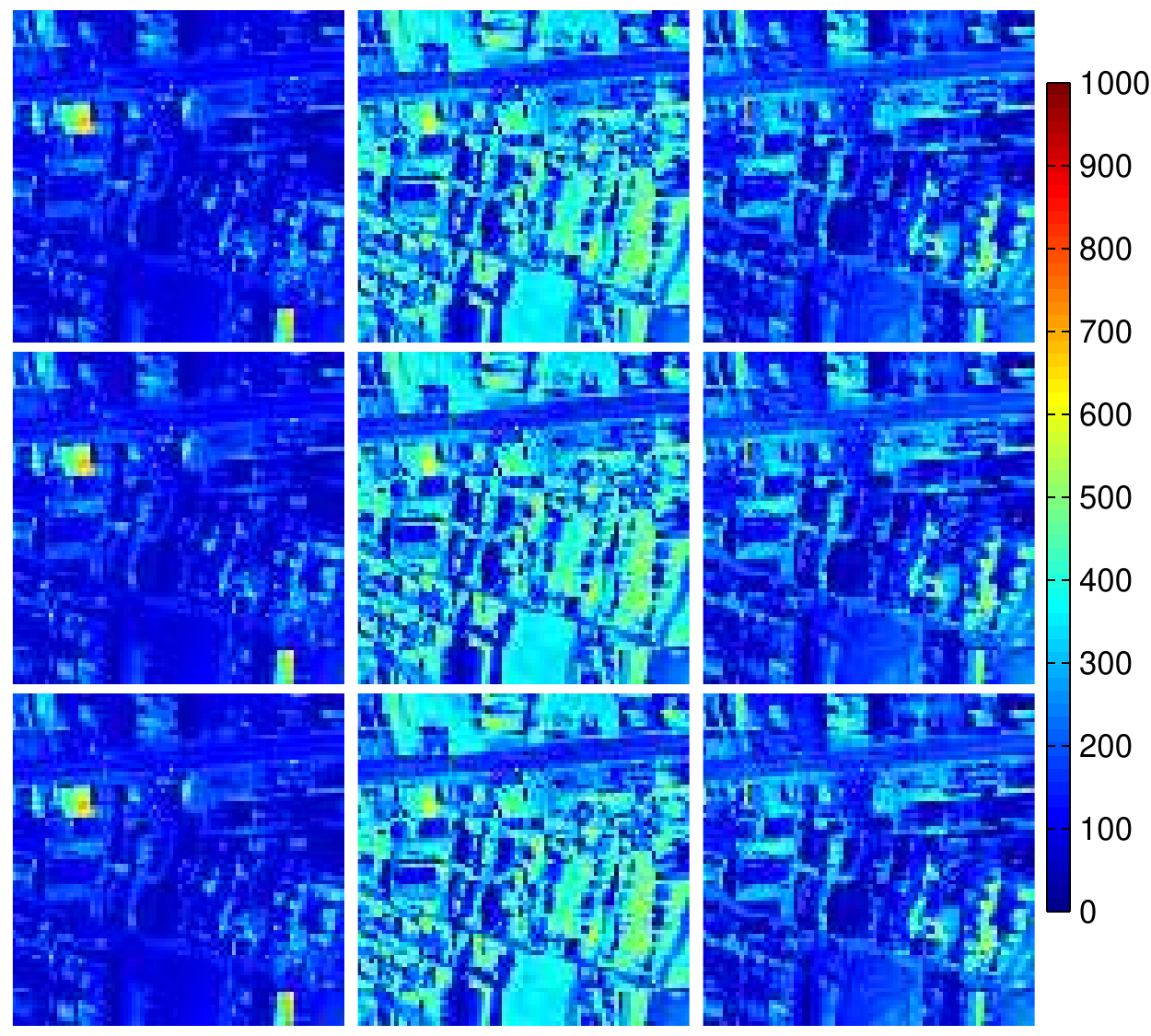}
\end{minipage}}
\end{minipage}
\end{figure}

\subsection{Tests on images}\label{subsec:imagetest}
Despite that natural image recovery from incomplete random samples is not a typical image processing task, we picked it to  test algorithm \ref{alg:ADM}, LMaFit, and FPCA since it is easy to visualize their solution qualities. This simulation used two grayscale images, the $768\times 1024$ Kittens and the $1200\times 1600$ Panda, shown in figure \ref{fig:cat-panda-orig}.

\begin{figure}[htbp]
\begin{center}
\includegraphics[width=0.45\textwidth]{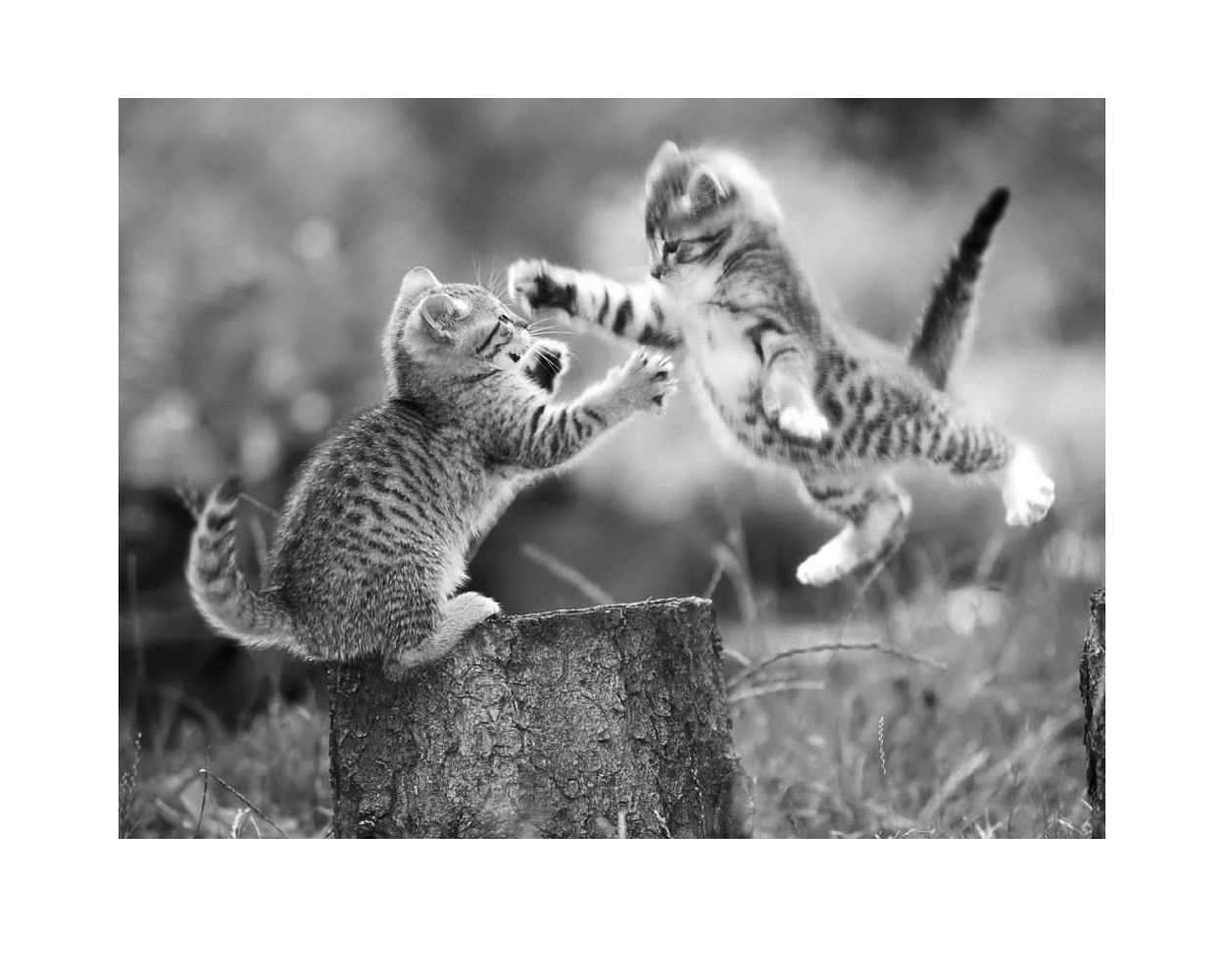}
\includegraphics[width=0.46\textwidth]{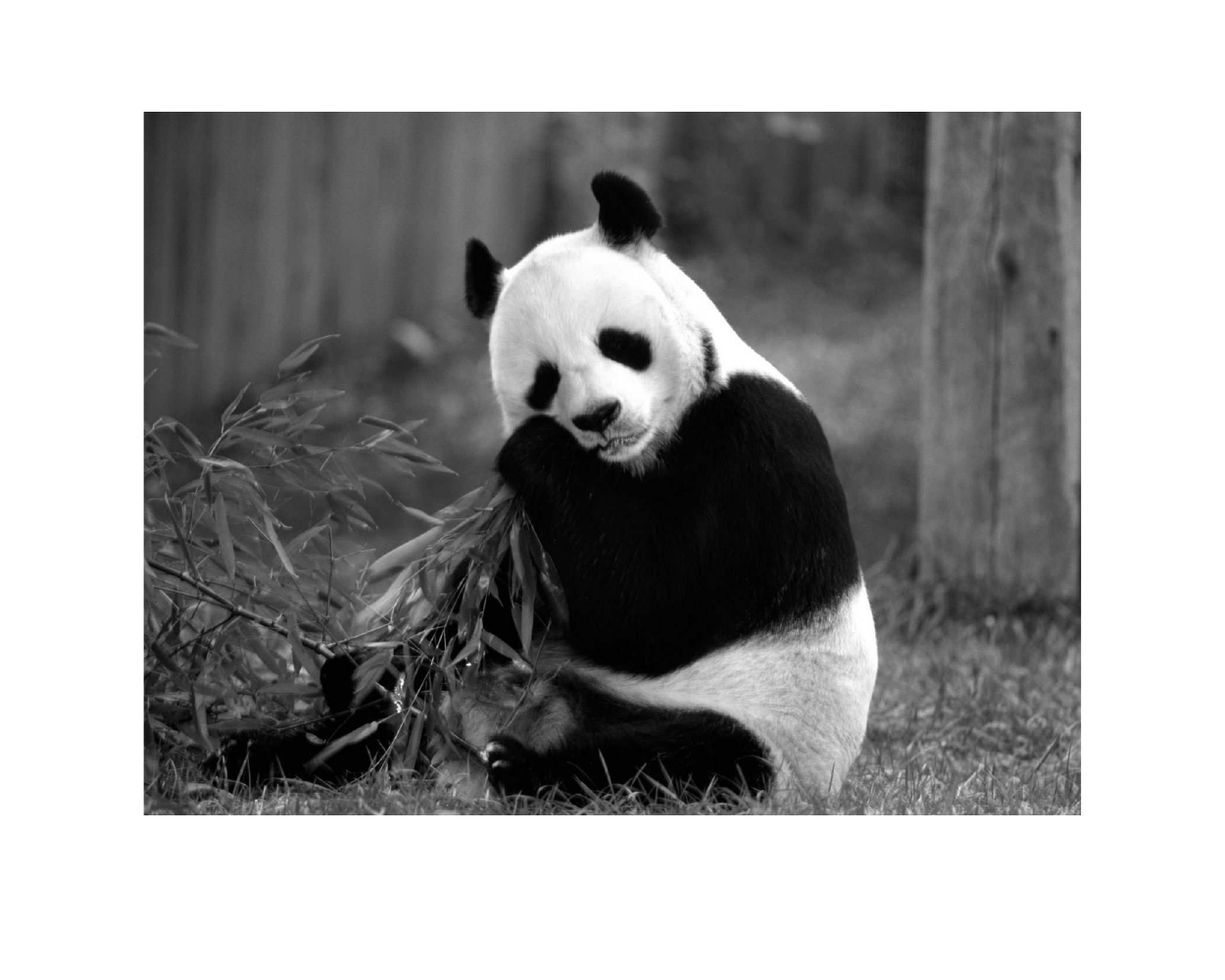}
\caption{Original  images: Kittens (left) and Panda (right)} \label{fig:cat-panda-orig}
\end{center}
\end{figure}

We applied relatively small (thus, challenging) sample rates of SR=10\%, 20\%, 30\% for Kittens and SR=10\%, 15\%, 20\% for Panda. We set tol=$10^{-5}$, and maxiter=2000 for Algorithm \ref{alg:ADM} and LMaFit and  est\_rank=2, rk\_inc =3 for  LMaFit. The parameters for FPCA were set to their default values.
 The results are given in tables \ref{table:cat} and \ref{table:panda} and the recovered images  in figures \ref{fig:cat-rec} and \ref{fig:panda-rec}.

Tables \ref{table:cat} and \ref{table:panda} indicate that FPCA performs slightly better than algorithm \ref{alg:ADM} in terms of recovery quality but slower when SR is as small as  10\% while at this SR, LMaFit performs much worse. With larger SRs such as 20\% and 30\% for Kittens and SR=15\% and 20\% for Panda, algorithm \ref{alg:ADM} is both faster and returns better images than FPCA. With SR=20\% and 30\%, algorithm \ref{alg:ADM} is better than LMaFit on Kittens in terms of recovery quality with comparable speed. However, with SR=20\%, LMaFit becomes slightly faster than algorithm \ref{alg:ADM} on Panda with comparable recovery quality. As SR further increases,  the three algorithms will return images with almost the same quality while  LMaFit is the best in speed.

\begin{table} \caption{Recover Kittens by Algorithm \ref{alg:ADM},
  LMaFit, and FPCA. The rank estimate for Algorithm \ref{alg:ADM} and LMaFit is 40. }\label{table:cat}
    {\scriptsize
 \setlength{\tabcolsep}{2pt}  %\centering
 \begin{center}
\begin{tabular}{|c|ccc|ccc|ccc|} \hline
\multicolumn{1}{|c|}{problem} & \multicolumn{3}{|c|}{Alg \ref{alg:ADM}}  &  \multicolumn{3}{|c|}{LMaFit} &\multicolumn{3}{|c|}{FPCA}   \\ \hline
seed & CPU & PSNR & MSE & CPU & PSNR & MSE & CPU & PSNR & MSE \\ \hline
\multicolumn{10}{|c|}{SR: 10\%}\\ \hline
  3445 &  20.98 & 18.21 & 1.51e-002  &  42.96 & 13.34 & 4.64e-002 &  23.01 & 20.09 & 9.80e-003 \\ \hline 
 31710 &  19.00 & 18.15 & 1.53e-002 &  17.34 & 14.58 & 3.48e-002  &  23.30 & 20.12 & 9.72e-003 \\ \hline 
 43875 &  20.03 & 18.07 & 1.56e-002 &  39.95 & 13.37 & 4.60e-002  &  23.40 & 20.18 & 9.59e-003 \\ \hline 
 69483 &  20.33 & 18.09 & 1.55e-002  &  13.60 & 15.14 & 3.06e-002 &  23.04 & 20.07 & 9.84e-003 \\ \hline 
 95023 &  20.06 & 18.04 & 1.57e-002  &  17.89 & 14.37 & 3.66e-002  &  23.25 & 20.06 & 9.87e-003 \\ \hline 
\hline
\multicolumn{10}{|c|}{SR: 20\%}\\ \hline
  3445 &  12.66 & 23.26 & 4.72e-003  &  11.88 & 21.50 & 7.08e-003  &  36.73 & 22.38 & 5.78e-003 \\ \hline 
 31710 &   9.68 & 23.15 & 4.84e-003  &   9.93 & 21.94 & 6.39e-003 &  33.92 & 22.25 & 5.95e-003 \\ \hline 
 43875 &   9.78 & 23.19 & 4.80e-003  &  13.01 & 21.37 & 7.30e-003 &  34.10 & 22.26 & 5.94e-003 \\ \hline 
 69483 &  10.15 & 23.12 & 4.87e-003  &  27.82 & 19.90 & 1.02e-002  &  34.20 & 22.38 & 5.78e-003 \\ \hline 
 95023 &   9.89 & 23.17 & 4.82e-003  &  10.20 & 21.49 & 7.10e-003 &  34.01 & 22.31 & 5.87e-003 \\ \hline 
\hline
\multicolumn{10}{|c|}{SR: 30\%}\\ \hline
  3445 &   9.63 & 24.53 & 3.53e-003  &   9.77 & 24.11 & 3.88e-003  &  54.65 & 23.44 & 4.53e-003\\ \hline 
 31710 &   7.98 & 24.48 & 3.56e-003  &   8.35 & 24.10 & 3.89e-003 &  54.64 & 23.30 & 4.68e-003 \\ \hline 
 43875 &   9.81 & 24.48 & 3.57e-003  &  12.54 & 24.03 & 3.95e-003 &  55.20 & 23.47 & 4.50e-003 \\ \hline 
 69483 &   7.91 & 24.45 & 3.59e-003  &   5.54 & 23.78 & 4.19e-003  &  54.83 & 23.40 & 4.57e-003 \\ \hline 
 95023 &   8.64 & 24.46 & 3.58e-003  &   9.33 & 24.06 & 3.92e-003  &  54.45 & 23.33 & 4.65e-003 \\ \hline 
\end{tabular}
  \end{center}
  }
 \end{table}

\begin{figure}
\caption{Recovered  $768\times 1024$ Kittens with estimate rank $q=40$  by Algorithm \ref{alg:ADM} (left), LMaFit (middle), and FPCA (right)}
\label{fig:cat-rec}
\begin{center}
\includegraphics[width=1\textwidth]{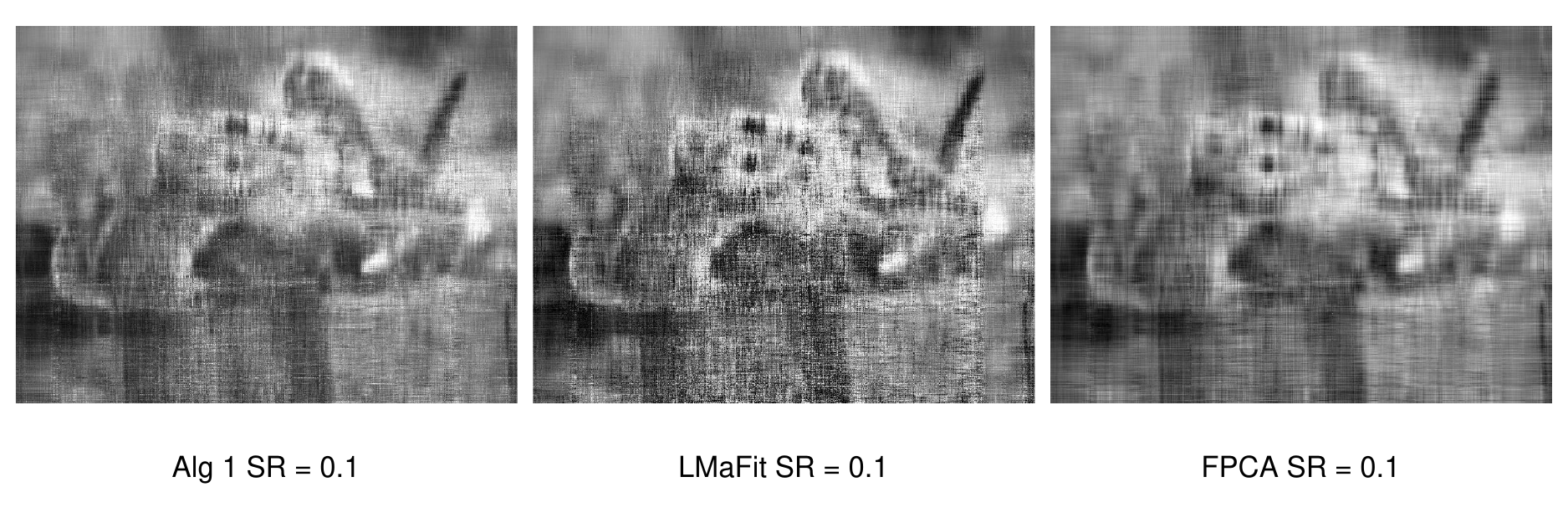}
\includegraphics[width=1\textwidth]{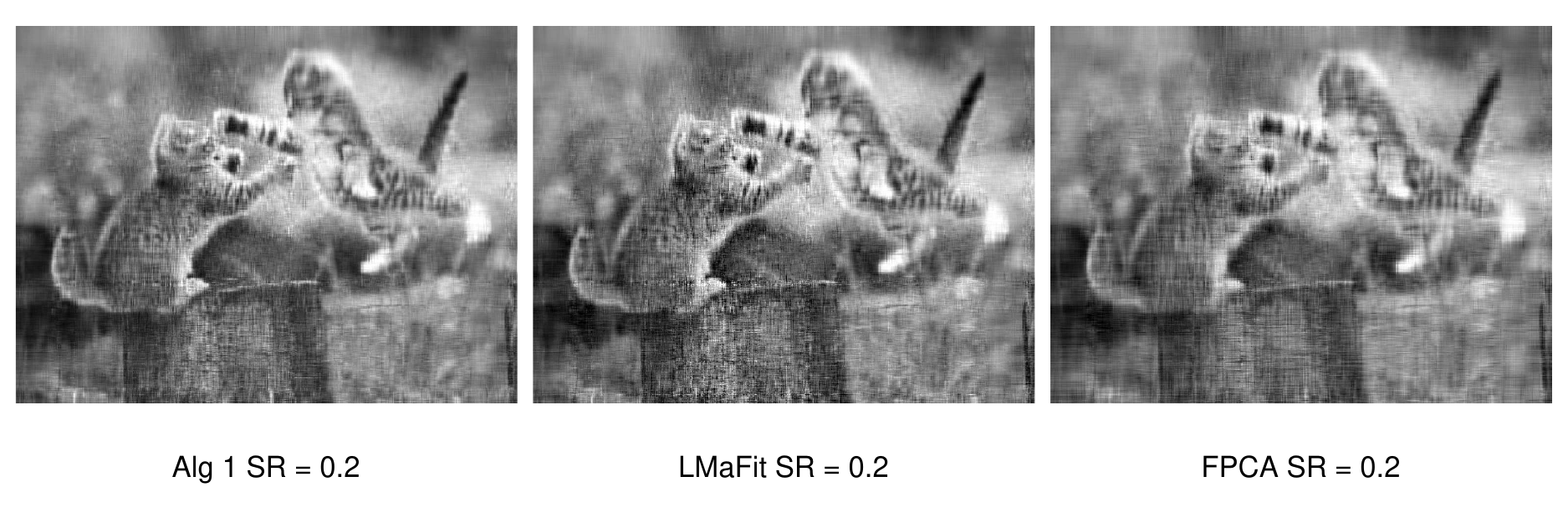}
\includegraphics[width=1\textwidth]{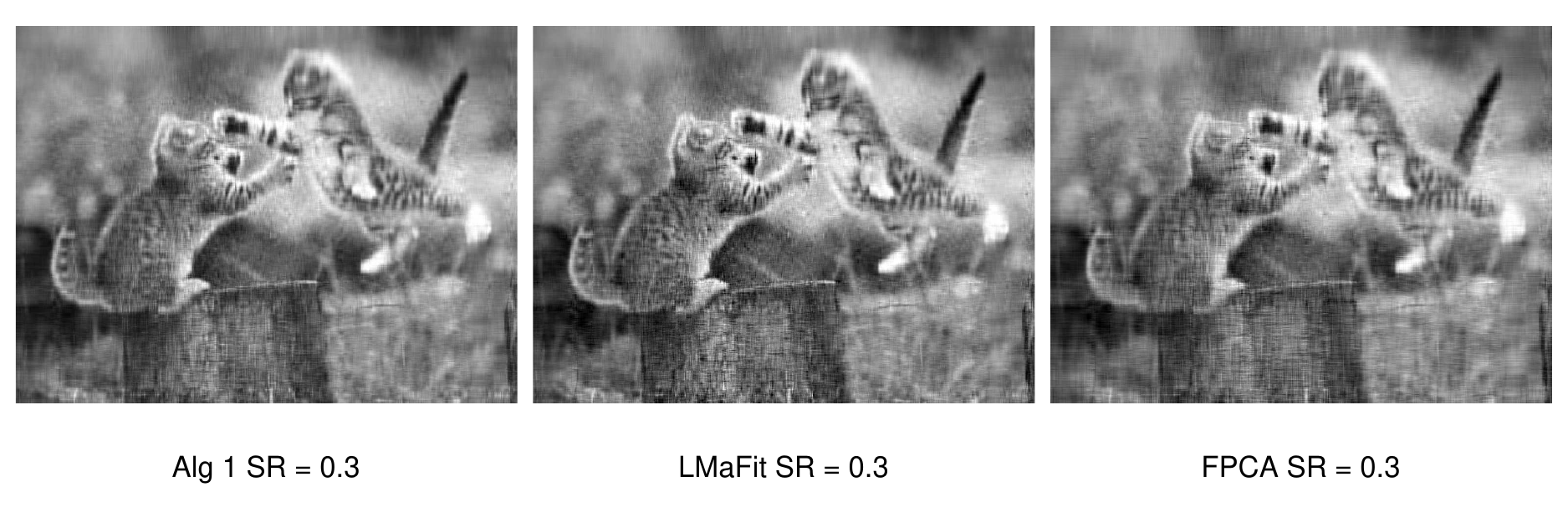}
\end{center}
\end{figure}

\begin{table} \caption{Panda: recovered images by Algorithm \ref{alg:ADM},
  LMaFit, and FPCA. The rank estimate for Algorithm \ref{alg:ADM} and LMaFit is 40. }\label{table:panda}
    {\scriptsize
 \setlength{\tabcolsep}{2pt}  %\centering
 \begin{center}
\begin{tabular}{|c|ccc|ccc|ccc|} \hline
\multicolumn{1}{|c|}{problem} & \multicolumn{3}{|c|}{Alg \ref{alg:ADM}}  &  \multicolumn{3}{|c|}{LMaFit} &\multicolumn{3}{|c|}{FPCA}   \\ \hline
seed & CPU & PSNR & MSE  & CPU & PSNR & MSE  & CPU & PSNR & MSE \\ \hline

\multicolumn{10}{|c|}{SR: 10\%}\\ \hline
  3445 &  45.40 & 23.57 & 4.40e-003  &  56.25 & 18.92 & 1.28e-002 &  59.07 & 23.51 & 4.46e-003 \\ \hline 
 31710 &  48.27 & 23.34 & 4.64e-003  &  66.37 & 18.55 & 1.40e-002  &  59.74 & 23.66 & 4.30e-003\\ \hline 
 43875 &  44.12 & 23.55 & 4.42e-003  &  50.24 & 19.46 & 1.13e-002  &  59.28 & 23.67 & 4.29e-003\\ \hline 
 69483 &  51.63 & 23.77 & 4.20e-003  &  27.67 & 20.77 & 8.37e-003  &  59.06 & 23.63 & 4.34e-003\\ \hline 
 95023 &  57.91 & 23.46 & 4.51e-003  &  33.41 & 19.67 & 1.08e-002 &  59.80 & 23.71 & 4.26e-003 \\ \hline 
\hline
\multicolumn{10}{|c|}{SR: 15\%}\\ \hline
  3445 &  32.97 & 25.83 & 2.61e-003  &  40.07 & 24.79 & 3.32e-003  &  73.32 & 25.26 & 2.98e-003\\ \hline 
 31710 &  29.51 & 25.74 & 2.67e-003  &  37.09 & 24.79 & 3.32e-003  &  72.35 & 25.15 & 3.06e-003\\ \hline 
 43875 &  32.57 & 25.84 & 2.61e-003  &  15.12 & 25.20 & 3.02e-003  &  72.60 & 25.14 & 3.06e-003\\ \hline 
 69483 &  29.34 & 25.80 & 2.63e-003  &  31.99 & 24.84 & 3.28e-003  &  72.85 & 25.13 & 3.07e-003\\ \hline 
 95023 &  26.69 & 25.80 & 2.63e-003  &  17.87 & 25.20 & 3.02e-003  &  71.65 & 25.03 & 3.14e-003\\ \hline 
\hline
\multicolumn{10}{|c|}{SR: 20\%}\\ \hline
  3445 &  28.45 & 26.55 & 2.21e-003  &  30.25 & 26.58 & 2.20e-003  &  91.00 & 25.81 & 2.62e-003\\ \hline 
 31710 &  25.79 & 26.54 & 2.22e-003  &  17.71 & 26.46 & 2.26e-003 &  91.23 & 25.87 & 2.59e-003 \\ \hline 
 43875 &  30.69 & 26.56 & 2.21e-003  &  20.51 & 26.27 & 2.36e-003  &  90.20 & 25.73 & 2.67e-003 \\ \hline 
 69483 &  23.13 & 26.56 & 2.21e-003  &  22.06 & 26.42 & 2.28e-003 &  90.40 & 25.85 & 2.60e-003 \\ \hline 
 95023 &  31.05 & 26.59 & 2.19e-003  &  27.80 & 26.58 & 2.20e-003 &  90.67 & 25.79 & 2.64e-003 \\ \hline 
\end{tabular}
  \end{center}
  }
 \end{table}

\begin{figure}
\caption{Recovered  $1200\times 1600$ Panda for estimate rank $q=40$  by Algorithm \ref{alg:ADM} (left), LMaFit (middle), and FPCA (right)}
\label{fig:panda-rec}
\begin{center}
\includegraphics[width=1\textwidth]{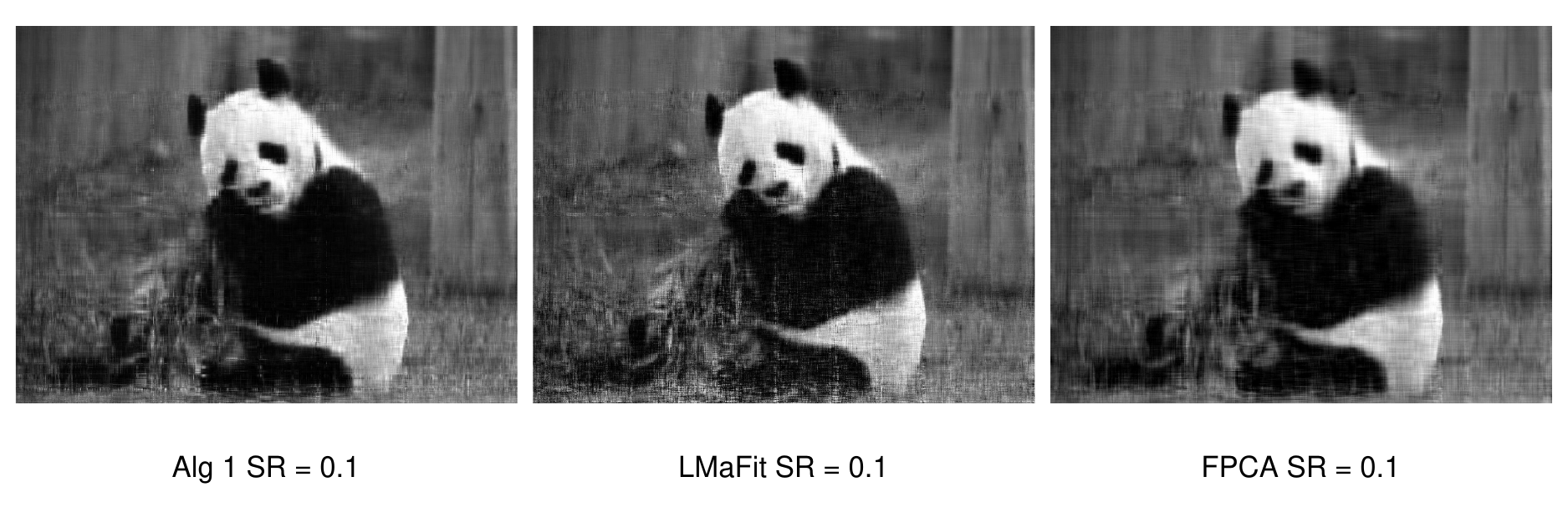}
\includegraphics[width=1\textwidth]{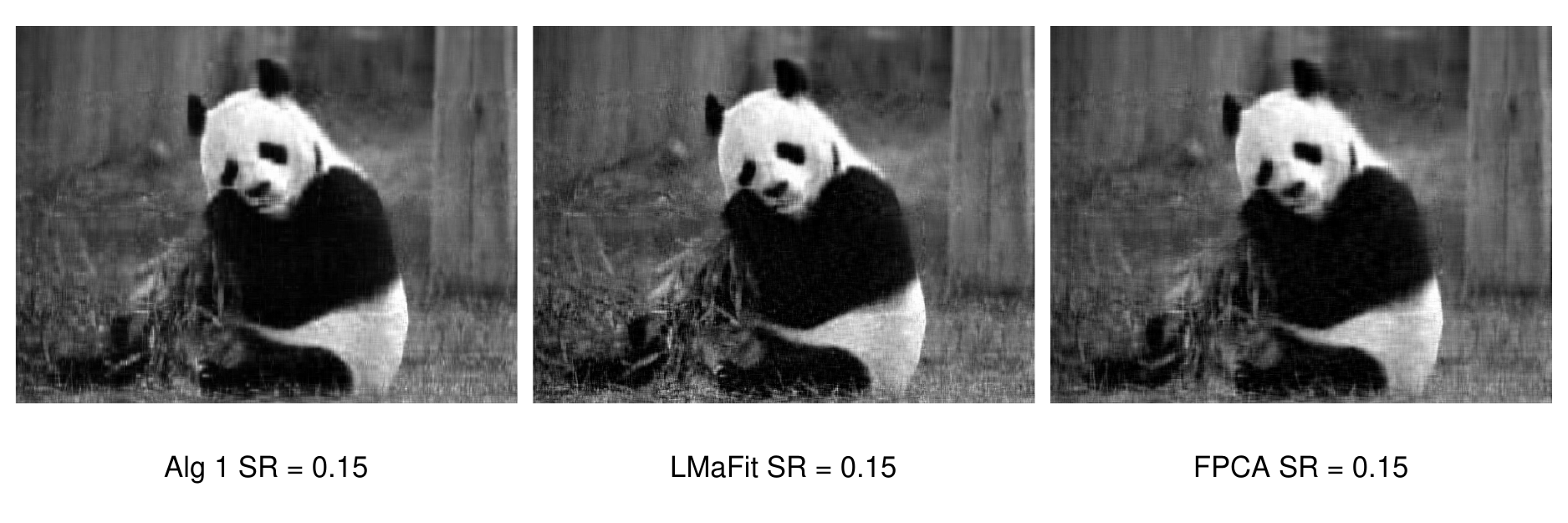}
\includegraphics[width=1\textwidth]{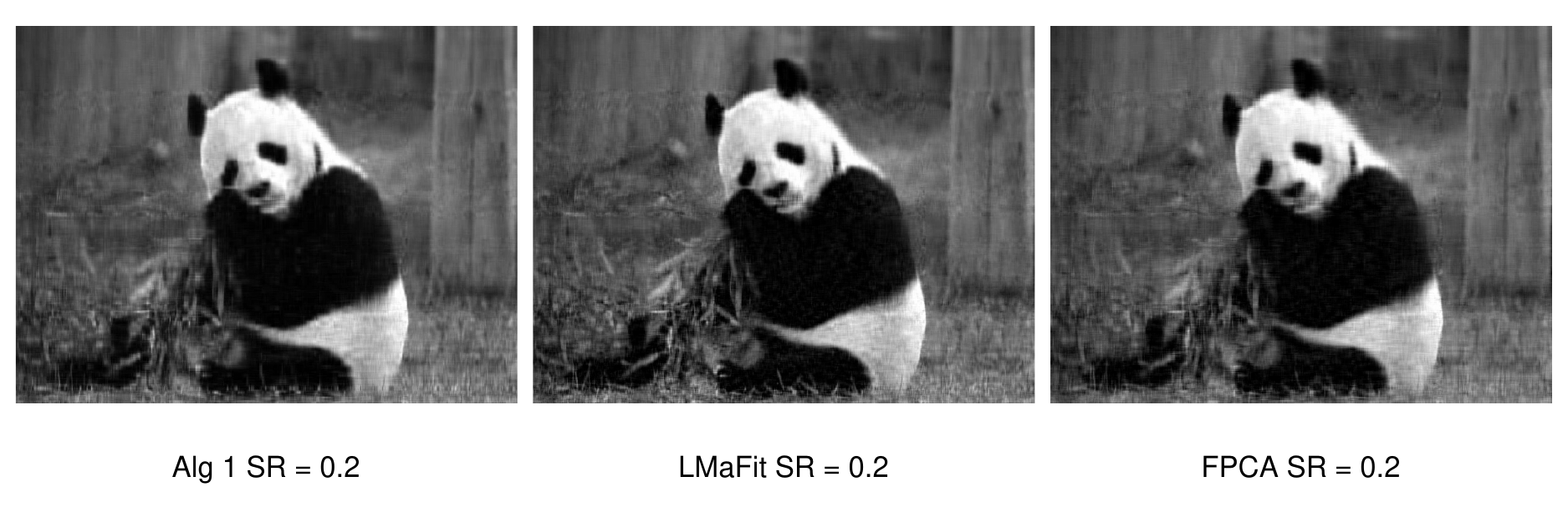}
\end{center}
\end{figure}

\section{Conclusions}
Among wide applications of nonnegative matrix factorization and those of low-rank matrix completion, there is a rich subset of problems where data matrices can be well approximated by matrix factorizations that are both low-rank and nonnegative, while some of the data (matrix elements) are missing.  To best recover missing data, we propose to combine nonnegative matrix factorization and matrix completion, utilizing both nonnegativity and low-rankness in a date recovery formulation.  This paper presents our first attempt to solve this non-convex formulation using an algorithm based on the classic alternating direction augmented Lagrangian method.  The algorithm has a relatively low per-iteration complexity, especially when the approximation rank is low.  Extensive numerical results in this paper indicate that the underlying formulation is useful, and the performance of the alternating direction algorithm is satisfactory.  Since global convergence and recovery guarantee results are still largely unknown, we hope that the results of this paper will also motivate further theoretical and numerical studies on this useful problem.

\section*{Acknowledgements} The work of W.~Yin was supported in part by US NSF CAREER Award
DMS-07-48839, ONR Grant N00014-08-1-1101, ARL and ARO grant W911NF-09-1-0383, and an Alfred P. Sloan Research Fellowship. The work of Z.~Wen was supported in part by NSF DMS-0439872 through
UCLA IPAM. The work of Y.~Zhang was supported in part by NSF DMS-0811188 and
ONR grant N00014-08-1-1101.

%\setcitestyle{numbers}
%\bibsep=0.05cm
%\bibliographystyle{siam}
\bibliographystyle{plain}
\bibliography{optimization,nmfc,inverse}
\end{document}